\documentclass[11pt]{article} 
\usepackage[utf8]{inputenc}
\usepackage[margin=1in]{geometry} 
\usepackage[numbers]{natbib}

\usepackage{booktabs} % For formal tables
\usepackage[ruled]{algorithm2e} % For algorithms
\usepackage{pifont}
\usepackage{nicefrac}
\usepackage{color}
\usepackage{xcolor,colortbl}
\usepackage{wrapfig}
\usepackage{graphicx}
\usepackage{amsmath,amssymb}
\usepackage{mathtools}
 
\usepackage{bm}
\usepackage{rotating}
\usepackage{multirow} 
\usepackage{enumitem}
\usepackage{hyperref}[backref=page]
\hypersetup{
     colorlinks   = true,
     linkcolor    = red, % color of internal links
     urlcolor     = blue, % color of external links
	 citecolor    = blue % color of links to bibliography
}

\newtheorem{thm}{Theorem}
\newtheorem{dfn}{Definition}

\newtheorem{ex}{Example}

\newtheorem{coro}{Corrollary}
\newenvironment{proof}{\noindent{\bf Proof:}}{\hfill $\Box$ }

\newtheorem{claimnew}{Claim}

\newcounter{newct}

\newtheorem{asmp}{Assumption}

\newcommand{\bv}{\begin{array}}

\newcommand{\ma}{\mathcal A}

\newcommand{\calP}{\mathcal P}
\newcommand{\calV}{\mathcal V}

\newcommand{\ml}{\mathcal L}

\newcommand{\calT}{\mathcal T}

\newcommand{\calS}{\mathcal S}
\newcommand{\calN}{\mathcal N}

\newcommand{\ra}{\rightarrow}

\newcommand{\hist}{\text{Hist}}

\newcommand{\myparagraph}[1]{\vspace{2mm}\noindent {\bf\boldmath #1}}

\newcommand{\Omit}[1]{}

\newcommand{\piuni}{{\pi}_{\text{uni}}}

\newcommand{\conv}{\text{CH}}

\newcommand{\wmg}{\text{WMG}}

\newcommand{\rev}[1]{{\text{Rev}\left(#1\right)}}

\newcommand{\CC}{\text{\sc CC}}

\newcommand{\Par}{\text{\sc Par}}
\newcommand{\CnP}{\text{\sc C\&P}}

\newcommand{\HM}{\text{\sc HM}}
\newcommand{\MM}{\text{\sc MM}}

\newcommand{\CnS}{\text{\sc C\&S}}
\newcommand{\SP}{\text{\sc SP}}
\newcommand{\CnM}{\text{\sc C\&M}}
\newcommand{\CnH}{\text{\sc C\&H}}

\newcommand{\satmax}[2]{\widetilde{#1}_{#2}^{\max}}
\newcommand{\satmin}[2]{\widetilde{#1}_{#2}^{\min}}
\newcommand{\sat}[1]{{#1}}

\newcommand{\flip}[1]{{\text{\rm flip}}({#1)}}

\newcommand{\KT}{\text{KT}}
 
\title{Semi-Random Impossibilities of Condorcet Criterion}
 
\author{ Lirong Xia\\ 
Rensselaer Polytechnic Institute, Troy, NY 12180, USA,\\
xialirong@gmail.com}
\date{}
 
\begin{document}
\maketitle              % typeset the header of the contribution

\begin{abstract}
The {\em Condorcet criterion} ($\CC$) is a classical and well-accepted criterion for voting. Unfortunately, it is incompatible with many other desiderata including {\em participation} ($\Par$), {\em half-way monotonicity} ($\HM$), {\em Maskin monotonicity} ($\MM$), and {\em strategy-proofness} ($\SP$). Such incompatibilities are often known as impossibility theorems, and are proved by worst-case analysis. Previous work has investigated the likelihood for these impossibilities to occur under certain models, which are often criticized of being unrealistic. 

We strengthen previous work by proving the first set of {\em semi-random impossibilities} for voting rules to satisfy $\CC$ and the more general, group versions of the four desiderata:  for any sufficiently large number of voters $n$, any size of the group $1\le B\le \sqrt n$, any voting rule $r$, and under a large class of {\em semi-random} models that include Impartial Culture, the likelihood for $r$ to satisfy $\CC$ and $\Par$, $\CC$ and $\HM$, $\CC$ and $\MM$, or $\CC$ and $\SP$ is $1-\Omega(\frac{B}{\sqrt n})$. This matches existing lower bounds for $\CnP$ ($B=1$) and $\CnS$ ($B\le \sqrt n$), showing that many commonly-studied voting rules  are already asymptotically optimal in such cases. 
\end{abstract}

\section{Introduction}
The {\em Condorcet criterion} of voting~\citep{Condorcet1785:Essai} is a classical desideratum that has  {\em ``nearly universal acceptance''}~\cite[p.~46]{Saari1995:Basic}. It requires  a voting rule to  choose the {\em Condorcet winner}---the alternative who beats other alternatives in head-to-head competitions---whenever it exists.  

Unfortunately, it is well-known that the {Condorcet criterion} ($\CC$ for short) is incompatible with many other desiderata (a.k.a.~{\em axioms}) when the number of alternatives $m$ is at least $3$. Such incompatibilities are often called {\em impossibility theorems}. For example, no voting rule satisfies 
\begin{itemize}
\item $\CC$ and {\em participation} ($\Par$ for short, which requires that no voter has incentive to abstain from voting), when $m\ge 4$~\citep{Moulin1988:Condorcets};
\item   $\CC$  and {\em half-way monotonicity} ($\HM$ for short, which requires that no voter has incentive to reverse his/her vote~\citep{Sanver2009:One-way});
%is prone to the {\em preference reversal paradox} when $m\ge 4$, where an agent has incentive to flip her vote to improve the outcome of voting~\cite{}. 
\item  $\CC$ and {\em Maskin monotonicity} ($\MM$ for short, which requires that any voter raising the position of  the winner relative to other alternatives does not change the winner~\citep{Maskin1999:Nash}), as a special case  of  the Muller-Satterthwaite theorem~(\citeyear{Muller77:SPA}); or 
\item  $\CC$ and {\em strategy-proofness} ($\SP$ for short, which requires that no agent has incentive to lie), as  a special case of the Gibbard-Satterthwaite theorem~(\citeyear{Gibbard73:Manipulation,Satterthwaite75:Strategy}).\end{itemize}
The four combinations of axioms are therefore denoted by  $\CnP$, $\CnH$, $\CnM$, and  $\CnS$, respectively. 
Proofs for these impossibility theorems are based on {worst-case analysis}, by identifying a single instance of violation. Therefore, they do not preclude the possible that such violations are rare in practice. Indeed, if so,   {``\em then one need not be unduly worried''}~\citep{Pattanaik1978:Strategy}. 

Studying how rare such impossibilities are in practice has been a popular and active field of research~\citep{Gehrlein2011:Voting,Diss2021:Evaluating}. 
%, and was often investigated by assuming a certain distribution of the votes, in particular the i.i.d.~uniform distribution, known as {\em Impartial Culture (IC)}. See, for example, books by~\citet{Gehrlein2011:Voting} and~\citet{Diss2021:Evaluating} for excellent surveys of the field. %The lower bounds under IC were commonly known as ``quantitative impossibility theorems'', for example on Gibbard-Satterthwaite.
Recently, the topic was investigated using {\em smoothed analysis}~\citep{Spielman2009:Smoothed,Baumeister2020:Towards,Xia2020:The-Smoothed}, which  can be viewed as a worst average-case analysis under {\em semi-random} models~\citep{Feige2021:Introduction}, following the frequentists' principle: the likelihood of violation of axioms is estimated under an adversarially chosen (i.e., worst-case) distribution for the votes from a given set  of distributions. For example, %Accurate upper bounds on such likelihood for commonly studied voting rules were obtained. For example, 
the  likelihood for $\CC$ or $\Par$ to be violated is $\Theta(\frac{1}{\sqrt n})$ for many voting rules for $n$ voters, under a large class of semi-random models~\citep{Xia2021:Semi-Random}. %, where $n$ is the number of voters. 

While this is good news, as violations vanish at a $\Theta(\frac{1}{\sqrt n})$ rate, they are not rare enough when the cost of violation is high. 
For example, if a violation of $\CC$ or $\Par$  leads to a revote, whose social cost is $\Theta(n)$, then the expected social cost is $\Theta({\sqrt n})$, which is non-negligible. As another example, if 
% if a violation of $\CC+\Par$  damages the reputation of the election authority in the sense that, assuming 
 everyone complaints on social media about the violation and gets $-1$ utility every time when seeing a complaint, then the social cost can be as high as $\Theta(n^2)$, meaning that the  expected social cost is $\Omega(n^{1.5})$, or in other words, $\Omega(\sqrt n)$ {\em per person}. In such situations, voting rules with rarer  violations are desirable.

But can any voting rule do better, and if so, by how much? The answer  lies in the lower bound on the likelihood of violations (under all rules), or equivalently, the upper bound on the likelihood of satisfying the axioms. In this paper, we address this question for the four combinations of axioms involving $\CC$ mentioned earlier, by proving {\em semi-random impossibilities}~\citep{Xia2020:The-Smoothed} under a large class of  semi-random models that are more general and  realistic  than the commonly-used i.i.d.~uniform distribution, known as the {\em Impartial Culture (IC)}. Therefore, the research question of this paper can be phrased as: %, therefore calling  such lower bounds {\em semi-random impossibilities}~\citep{Xia2020:The-Smoothed}, and asking the following question for the four combinations of desiderata involving $\CC$ mentioned in the beginning of this paper. %term such lower semi-random But can we do better? Or in other words, 
\begin{center}
{\bf What are the semi-random impossibilities of $\CC$?}
\end{center}%limitations of any voting rules on satisfying these impossibilities under semi-random models?
%We address this question for $\CC$ together with the more general, group versions of the four desiderata: $\CC$ and participation (denoted by $\CnP$), $\CC$ and avoidance of preference reversal paradox (denoted by $\CnH$), $\CC$ and Maskin monotonicty (denoted by $\CnM$), and $\CC$ and strategy-proofness (denoted by $\CnS$). %In fact, we will prove the theorems for the stronger group versions of these properties.  
More precisely, we consider  the more general, group versions of $X\in \{\CnP, \CnH, \CnM,\CnS\}$. %, defined as follows. 
For any $B\ge 1$, any collection of votes $P$ (called a {\em profile}), and any voting rule $r$, we let $\sat{X}(r,P,B)=1$ if $r$ satisfies $\CC$ at $P$ and no group of at most $B$ voters in $P$ can collaboratively violate $X$; otherwise $\sat{X}(r,P,B)=0$. %The classical properties correspond  to the $B=1$ case.
Then, given a set of distributions $\Pi$ over the votes and $n$ agents, the {\em semi-random version of $X$}~\citep{Xia2020:The-Smoothed,Xia2021:Semi-Random} is defined as: %denoted by $\satmin{X}{\Pi}$ and is defined as follows~\citep{Xia2020:The-Smoothed,Xia2021:Semi-Random}.
%on the simultaneous satisfaction of $\CC$ and $\Par$ under the {\em smoothed social choice framework}~\citep{Xia2020:The-Smoothed}, which is inspired by the smoothed complexity analysis~\citep{Spielman2004:Smoothed} and is more general than IC or i.i.d.~distributions in general. %The framework is based on the worst average-case analysis as in the classical smoothed complexity analysis~\citep{Spielman2004:Smoothed}. 
%We adopt the semi-random model in~\citep{Xia2020:The-Smoothed}, where an adversary chooses a (worst-case) distribution for each agent from a set of given distributions $\Pi$, and then the satisfaction of a property $X$ is evaluated on the randomly generated profiles. More precisely, given a set of $m\ge 3$ alternatives, denoted by $\ma$, and a set of distributions $\Pi$ over linear orders over $\ma$, denoted by $\ml(\ma)$. There are $n\ge 1$ agents. Each agent's preferences are modeled by a random variable whose distribution is in $\Pi$. 
%For every $X\in \{\CnP, \CnH, \CnM,\CnS\}$, every $B\ge 1$, and every voting rule $r$, the 
%smoothed likelihood of $X$, denoted by $\tilde{X}$, is defined as follows~\citep{Xia2020:The-Smoothed,Xia2021:Semi-Random}.
\begin{equation}
\label{dfn:s-sat}
\satmin{X}{\Pi}(r,n,B) \triangleq \inf\nolimits_{\vec\pi\in\Pi^n}\Pr\nolimits_{P\sim\vec\pi}\left(\sat{X}(r,P,B) =1\right)
\end{equation}
That is, $\satmin{X}{\Pi}(r,n,B)$ is the worst-case (lower bound) on the probability for $X$ to be $1$ under the profile $P$ generated from a  vector $\vec\pi$ of $n$ distributions in $\Pi^n$, one for each agent. Notice that while agents' votes are independently generated, their underlying distributions are adversarially chosen and can be different. A high $\satmin{X}{\Pi}$ value is desirable, because it  implies that the expected satisfaction of $X$ is high even under the worst distribution $\vec \pi\in \Pi^n$.  
%In other words,  imagine that  an adversary   tries to minimize the value of $X$, but can only choose the distribution of preferences for each agent, instead of their (deterministic) preferences as in the classical worst-case analysis. Then, $\tilde{X}$ measures the minimum (expected) value the adversary can achieve. When $X$ represents a desirable property, a low $\tilde{X}$ value  is negative news, because it states that the adversary can choose agents' ``ground truth'' preferences, so that $X$ holds with small probability.

%It is not hard to see that smoothed analysis generalizes the worst-case analysis (where $\Pi$ contains deterministic distributions for all linear orders) and the i.i.d.~case (where $|\Pi|=1$), though its practical relevance largely depends on the choice of $\Pi$. In this paper, we prove  the following smoothed impossibility theorem for the simultaneous satisfaction of $\CC$ and $\Par$, i.e., $X=\nCP$, under mild assumptions on $\Pi$.

%The significance of addressing this question is two-fold. First, theoretically, such semi-random impossibilities (lower bounds) help us understand the limitations of voting rules, and in particular, how much room for improvement we may have.  Second, practically, this will allow policy makers to make informed decisions, as to choose the ``optimal'' voting rule. After all, the goal is to provide theoretical comparisons, arguably a ``realistic'' one.

\subsection{Our Contributions}
\label{sec:contributions}
%We answer the question for the more general, group versions of the four combinations of properties discussed in the beginning of this paper under the semi-random model in~\citep{Xia2020:The-Smoothed}, where an adversary chooses a (worst-case) distribution for each agent from a set of given distributions $\Pi$, and then the satisfaction of a property $X$ is evaluated on the randomly generated profiles.

%\paragraph{Main results.}  
The main results of this paper are  four semi-random impossibility theorems: for $X = \CnP$ (Theorem~\ref{thm:CC+Par}), $\CnH$ (Theorem~\ref{thm:CC+HM}), $\CnM$ (Theorem~\ref{thm:CC+MM}), or $\CnS$ (Theorem~\ref{thm:CC+SP}), any $m\ge 4$ ($m\ge 3$ for $\CnM$ and $\CnS$), any sufficiently large $n$, any $1\le B\le \sqrt n$,  any voting rule $r$, and any $\Pi$ satisfying certain conditions (Assumption~\ref{asmp:strict}), 
$$\satmin{X}{\Pi}(r,n,B) = 1- \Omega\left(\frac{B}{\sqrt n}\right)$$ 
In other words,   no voting rule can guarantee that $X$ is violated with probability smaller than  $\Omega(\frac{B}{\sqrt n})$.   %As a corollary, all semi-random impossibilities hold under IC~(Corollary~\ref{coro:IC}).
The results also imply that every additional member in the group (up to $\sqrt n$) roughly increase the likelihood  of violation by $\Theta(\frac{1}{\sqrt n})$. Specifically, when $B=\Omega(\sqrt n)$,  the likelihood of violation does not vanish even in large elections ($n\ra\infty$). %, where the constant only depends on $m$ and $\Pi$ and does not depend on $B$, $n$, or $r$.

Our results match the lower bound for $\CnP$ when $B=1$~\citep{Xia2021:Semi-Random} and  for $\CnS$ for every $B\le \sqrt n$~\citep{Xia2022:How-Likely}, which are achieved by many  voting rules that satisfies $\CC$, such as Copeland, maximin, ranked pairs, and Schulze---in contrast, for $\CnP$, positional scoring rules and STV are much worst, as their satisfactions are $1-\Theta(1)$~\citep{Xia2021:Semi-Random}. 

\myparagraph{Good  or bad news?} On the positive side,  it is the first time, to the best of our knowledge, that the optimal likelihood of avoiding impossibility theorems that involve $\CC$ is known. It is  surprising to us that many existing rules are already optimal. %These results shed some light to the Borda vs.~Condorcet debate in favor of rules that satisfy $\CC$, because the Borda rule .
%Therefore, on the positive side,  such optimality/tightness results suggest that many existing voting rules are already  asymptotically optimal , so that the mechanism designer can focus on other desiderata. 
On the negative side, the tightness  suggests that there is little room for improvement, which can be a critical concern when the cost of violation is high. After all, we believe that these semi-random impossibility theorems are useful and informative  in theory, as they reveal limitations of the optimal rules, as well as in practice, for the decision maker to choose the voting rule   and decide the policies when a violation of axioms occurs. 

\myparagraph{Generality and limitations.}  The generality of the semi-random impossibilities proved in this paper largely depends on the  restrictiveness of Assumption~\ref{asmp:strict}. We defer the formal technical definition and  discussions to Section~\ref{sec:prelim}, and feel that the assumption is mild in practice, because it is satisfied by many single-agent preference models, including IC, the single-agent Mallows  and single-agent Plackett-Luce with bounded parameters~\citep{Xia2020:The-Smoothed}. As a result, the $1- \Omega (\frac{B}{\sqrt n} )$ upper bound  naturally holds under IC~(Corollary~\ref{coro:IC}). 

The major limitations are, first, the constant in $\Omega (\frac{B}{\sqrt n} )$ may be exponentially large in $m$, though it does not depend on $n$, $B$, or $r$. Second, the semi-random model in this paper assumes that the votes are statistically independent (but not necessarily identically distributed). These are common limitations/assumptions in preference modeling, see, e.g.,~\citep{Thurstone27:Law,Berry95:Automobile,Train09:Discrete}. Addressing them may require  breakthroughs in probability theory and are important and  challenging directions. %See  Section~\ref{sec:prelim} for more discussions.

\paragraph{\bf Proof overview.} The high-level idea is surprisingly simple: for each $X$ studied in this paper, in step 1, we leverage existing proof  of the (worst-case) impossibility theorem  to identify sufficiently many profiles where $X$ is violated. Then, in step 2, we prove that there exists $\vec\pi\in\Pi^n$ under which with $\Omega(\frac{B}{\sqrt n})$ probability, a profile falls in the set identified in  step 1.

Nevertheless, the actual calculations are technical challenging due to the generality of $r$. In step 1, we introduce a {\em rotated template} by scaling up an existing proof diagram (e.g.,~\citep[Chapter 1]{Peters2019:Fair}) to identify profiles where $X$ is violated, and prove that  there are sufficiently many such violation profiles by upper-bounding the number of times each of them is identified by the rotated template. Then in step 2, we use an averaging argument over all $n!$ permutations of a carefully chosen $\vec\pi$ to convert the problem to the likelihood about the histogram of profiles, which is then tackled  by  applying the point-wise concentration bound~\citep[Lemma~1]{Xia2021:How-Likely}. 

The  idea  and  techniques  have the potential to leverage other (worst-case)  impossibility theorems to  their semi-random versions. See Section~\ref{sec:other-imp} for more discussions.

%The main challenge in step 1 is to identify sufficiently many ``violation profiles'' in a neighborhood of an ``expected profile'' of some $\vec\pi\in\Pi^n$, yet we do not have much information about the violation profiles, as the voting rule $r$ can be arbitrary. This is addressed by creating {\em violation templates} that scale up existing proof diagrams, e.g.,~\citep[Chapter 1]{Peters2019:Fair} by a factor of $\sqrt n$, applying the violation templates to certain profiles, and then using a flipped diagram to prove that there are sufficiently many violation profiles as desired. The main challenge in step 2 is that it is unclear how such $\vec\pi$ can be chosen, because the violation profiles identified in step 1 can be quite arbitrary. This is addressed by using an averaging argument over all $n!$ permutations of a carefully chosen $\vec\pi$, converting the likelihood to the likelihood about the histogram of profiles, and then applying the point-wise concentration bound~\citep[Lemma~1]{Xia2021:How-Likely} to provide the lower bound on the likelihood of violation. 

\subsection{Related Work and Discussions}
\label{sec:related}
\noindent{\bf Condorcet criterion ($\CC$)} %The axiom was proposed by~\citeauthor{Condorcet1785:Essai} in 1785~\citep{Condorcet1785:Essai} and has  {\em ``nearly universal acceptance''}~\cite[p.~46]{Saari1995:Basic}. 
is satisfied by many commonly-studied voting rules. Prominent exceptions are positional scoring rules~\citep{Fishburn74:Paradoxes}  and multi-round-score-based elimination rules, such as STV.   Much previous work aimed at theoretically characterizing the {\em Condorcet efficiency}, which is the probability for the Condorcet winner to win  conditioned on its existence~\citep{Fishburn1974:Simple,Fishburn1974:Aspects,Paris1975:Plurality,Gehrlein1978:Coincidence,Newenhizen1992:The-Borda}. %The study can be dated back to~\citeauthor{Fishburn1974:Simple}'s  computer simulations~\citep{Fishburn1974:Simple,Fishburn1974:Aspects}. Analytical work was pioneered by~\citet{Paris1975:Plurality}, and  \citet{Gehrlein1978:Coincidence} proved that when there are three alternatives ($m=3$), Borda maximizes the asymptotic Condorcet efficiency under IC  among all positional scoring rules. This result was generalized to arbitrary $m$ by~\citet{Newenhizen1992:The-Borda}. 
%Computer simulations were used to study the likelihood of $\CC$ beyond positional scoring rules, see, e.g., \citep{Fishburn1976:An-analysis,Fishburn1977:An-analysis,Merrill1985:A-statistical,Nurmi1992:An-Assessment}.  %The semi-random satisfaction of $\CC$ for commonly-studied voting rules was investigated recently~\citep{Xia2021:Semi-Random}.
 
\myparagraph{\bf\boldmath Participation ($\Par$)}  was introduced to study voting rules that avoid the  {\em no-show paradox}~\citep{Fishburn1983:Paradoxes}. %In other words, $\Par$ requires that no voter has incentive to abstain from voting which holds under all positional scoring rules. 
\citet{Moulin1988:Condorcets} proved that when $m\ge 4$ and $n\ge 25$, no voting rule satisfies  $\CC$ and $\Par$ simultaneously. The bound on $n$ was characterized to be $12$ by simplified, SAT-solver-based proofs~\citep{Brandt2017:Optimal,Peters2019:Fair}. %Specifically, \citet{Brandt2017:Optimal} and~\citet{Peters2019:Fair} used SAT solvers to obtain an elegant proof diagram (Figure~\ref{fig:proof}), which will be used in our proofs. 
The likelihood of $\Par$ satisfaction by popular voting rules under IC was %proposed as an open question by~\citet{Berg1994:On-probability} and was 
investigated  in a series of work as summarized by~\citet[Chapter~4.2.2]{Gehrlein2011:Voting}, and also more recently by~\citet{Brandt2021:Exploring}.
%including~\citep{Berg1994:On-probability,Lepelley1996:The-likelihood,Lepelley2001:Scoring,Wilson2007:Probability,Brandt2021:Exploring}, also see~\cite[Chapter~4.2.2]{Gehrlein2011:Voting}. %The semi-random satisfaction of $\Par$ for commonly-studied voting rules was investigated recently~\citep{Xia2021:Semi-Random}, showing an $1-O(\frac{1}{\sqrt n})$ upper bound for many commonly studied voting rules that satisfy $\CC$. %,  the semi-random satisfaction of $\Par$ is $1-O(\frac{1}{\sqrt n})$.

\myparagraph{Half-way monotonicity ($\HM$)}  was introduced to study voting rules that avoid the {\em preference reversal paradox}, and was proved to be incompatible with $\CC$~\citep{Sanver2009:One-way}. \citet{Peters2017:Condorcets} used SAT solvers to characterized the number of voters under which the impossibility holds. 

\myparagraph{Maskin monotonicity ($\MM$)} was introduced to characterize Nash implementability~\citep{Maskin1999:Nash}. The Muller-Satterthwaite theorem~\citep{Muller77:SPA} establishes the equivalence between $\MM$ and $\SP$ in the worst-case sense: a voting rule satisfies $\MM$ if and only if it satisfies $\SP$. 

\myparagraph{Strategy-proofness ($\SP$)} cannot be satisfied by any non-dictatorial and unanimous voting rules when $m\ge 3$, due to the Gibbard-Satterthwaite theorem~\citep{Gibbard73:Manipulation,Satterthwaite75:Strategy}.  $\SP$ is stronger than $\HM$, because the latter uses a special form of manipulation (by reversing the truthful vote). At a high-level, $\Par$ can be viewed as a weak form of $\SP$ that prevents  manipulation by abstention, though $\Par$ is not  weaker than $\SP$ by definition, because $\Par$ reasons about elections of different sizes. 
A quantitative Gibbard-Satterthwaite theorem (under IC) was proved for $m=3$ by~\citet{Friedgut2011:A-quantitative}, and was subsequently developed in~\citep{Dobzinski08:Frequent,Xia08:Sufficient,Isaksson10:Geometry}, and the case for general $m$ was resolved by~\citet{Mossel2015:A-quantitative}. 

\myparagraph{\boldmath Semi-random $\CnP$, $\CnH$, $\CnM$, and $\CnS$.} We are not aware of any semi-random impossibility theorem  about the satisfaction of $\CnP$, $\CnH$, or $\CnM$, even under IC. For $\SP$, the quantitative Gibbard-Satterthwaite theorem by~\citet{Mossel2015:A-quantitative}  establishes an $1-\Omega(\frac{1}{n^{67}})$ upper bound under IC for any voting rule that is sufficiently different from  dictatorships. Therefore, the same bound holds for $\CnS$ for any rule that satisfies $\CC$. The $1-\Omega(\frac{B}{\sqrt n})$ upper bound for $\CnS$ in our Theorem~\ref{thm:CC+SP} also applies to all $\CC$ rules, which is stronger than the special case of~\citep{Mossel2015:A-quantitative}, because our bound is lower and works for every $B\le \sqrt n$  under   more general  models.

For possibility results (i.e., lower bound for  optimal rules), as discussed in Section~\ref{sec:contributions}, 
%a  $1-\Theta(\frac{1}{\sqrt n})$ lower bound was known for $\CnP$ for $B=1$~\citep{Xia2020:The-Smoothed}, and a  $1-\Theta(\frac{B}{\sqrt n})$ lower bound was known for $\CnS$ for every $B\le \sqrt n$~\citep{Xia2022:How-Likely}. Both bounds are achieved by many commonly-studied voting rules under semi-random models that satisfy Assumption~\ref{asmp:strict}. Therefore, for all four combinations of axioms studied in this paper, no tight bound was previously known even under IC. 
our results imply that the bounds are tight for $\CnP$ (when $B=1$) and for $\CnS$ (when $B\le \sqrt n$). We conjecture that they are tight for other axioms studied  in this paper with all $B\le \sqrt n$. %Compared to the $1-\Omega(\frac{1}{n^{67}})$ lower bound for $\CnS$~\citep{Mossel2015:A-quantitative},  

%beyond the special case of IC or any quantitative versions (i.e., lower bound) of the satisfaction of  $\CC$+$\Par$ impossibility, the existence of preference reversal paradox for $\CC$ rules, or the violation of Maskin monotonicity under $\CC$ rules. Technically, the proof of   quantitative impossibility theorems, which are special cases of semi-random impossibilities, can already be complicated and challenging~\citep{Friedgut2011:A-quantitative,Mossel2015:A-quantitative}. As a result, for $\CC$ rules, the likelihood of violation of strategy-proofness is $\Omega(\frac{1}{n^{67}})$~\citep{Mossel2015:A-quantitative} under IC, but it is unclear whether this lower bound is tight, as an upper bound of $O(\frac{1}{\sqrt n})$ is known for many voting rules~\citep{Xia2022:How-Likely}.

\myparagraph {\bf Quantitative and semi-random impossibilities.} There is a large body of literature on quantitative  impossibility theorems in social choice under IC. For example,   quantitative  versions of Arrow's impossibility theorem~\citep{Arrow63:Social} were proved~\citep{Kalai02:Fourier,Mossel2012:A-quantitative,Keller2012:-A-tight,Mossel2013:On-reverse}. In judgement aggregation, \citet{Nehama2013:Approximately} and~\citet{Filmus2020:AND-testing} developed quantitative characterizations of AND-homomorphism as oligarchy, whose worst-case version was due to~\citet{List2002:Aggregating,List2004:Aggregating}. \citet{Xia2020:The-Smoothed} proved a semi-random version of the ANR impossibility theorem on {\em anonymity} and {\em neutrality}, whose worst-case version was proved by~\citet{Moulin1983:The-Strategy}.

\myparagraph {\bf Other smoothed/semi-random results.} Semi-random models have been widely adopted to analyze the performance of algorithms in practice in  combinatorial optimization~\citep{Blum1995:Coloring}, mathematical programming~\citep{Spielman2004:Smoothed}, machine learning~\citep{Blum2002:Smoothed}, and  algorithmic game theory \citep{Chung2008:The-Price,Psomas2019:Smoothed,Boodaghians2020:Smoothed,Blum2021:Incentive-Compatible}, etc. %\citep{blum1995coloring,feige1998heuristics,Spielman2004:Smoothed,kolla2011play}. 
We refer the readers to recent surveys on semi-random models~\citep{Feige2021:Introduction} and general approaches beyond worst-case analysis~\citep{Roughgarden2021:Beyond}.  In addition to the work discussed above, semi-random/smoothed analysis has been applied to other  social choice problems, e.g., likelihood of ties~\cite{Xia2021:How-Likely}, complexity of winner determination~\citep{Xia2021:The-Smoothed}, judgement aggregation~\citep{Liu2022:Semi-Random}, and fair division~\citep{Bai2022:Fair}. 
 
\section{Preliminaries}
\label{sec:prelim}
%\noindent{\bf Basic Setting.} 
For any  $q\in\mathbb N$, we let $[q]=\{1,\ldots,q\}$. Let $\ma=[m]$ denote the set of $m\ge 3$ {\em alternatives}. Let $\ml(\ma)$ denote the set of all linear orders over $\ma$. Let $n\in\mathbb N$ denote the number of voters (agents). Each voter uses a linear order $R\in\ml(\ma)$ to represent his or her preferences, called a {\em vote}, where $a\succ_R b$ means that the agent prefers alternative $a$ to alternative $b$. The vector of $n$ voters' votes, denoted by $P$, is called a {\em (preference) profile}, sometimes called an $n$-profile. %The set of $n$-profiles for all $n\in\mathbb N$ is denoted by $\ml(\ma)^* = \bigcup_{n =1}^{\infty} \ml(\ma)^n$.  
A voting rule $r$ maps any profile to  a single winner.  %A {\em fractional} profile is  a   profile $P$ coupled with a possibly non-integer and/or negative weight vector $\vec \omega_P=(\omega_R:R\in P)\in{\mathbb R}^{n}$ for the votes in $P$.  Sometimes  the weight vector is omitted when it is clear from the context or when $\vec\omega_P=\vec 1$. 
For any  profile $P$, let $\hist(P)\in {\mathbb R}_{\ge 0}^{m!}$ denote the anonymized version of $P$, also called the {\em histogram} of $P$, which contains the total number of each linear order in $\ml(\ma)$ according to $P$. %A voting rule maps any profile to a set of alternatives, called the {\em winners}. % A  {\em resolute}  voting rule $r$ maps any profile to a set that consists of a single winner. 

\vspace{2mm}\noindent{\bf Weighted majority graphs and the Condorcet winner.}  For any profile $P$ and any pair of alternatives $(a,b)$, let $ P[a\succ b]$ denote the  number of votes in $P$ where $a$ is preferred to $b$. Let $\wmg(P)$ denote the {\em weighted majority graph} of $P$, whose vertices are $\ma$ and whose weight on edge $a\ra b$ is $w_P(a,b) = P[a\succ b] - P[b\succ a]$. %Let $\umg(P)$ denote the  {\em unweighted majority graph}, which is the unweighted directed graph that is obtained from  $\wmg(P)$ by keeping the edges with strictly positive weights. %Sometimes a distribution $\pi$ over $\ml(\ma)$ is viewed as a fractional profile, where for each $R\in\ml(\ma)$ the weight on $R$ is $\pi(R)$. In such cases, we let $\wmg(\pi)$ denote the weighted majority graph of the {fractional} profile represented by $\pi$.  
The {\em Condorcet winner} of a profile $P$ is the alternative whose outgoing edges in $\wmg(P)$ are positively weighted. %Let $\cwinner(P)$  denote the set of Condorcet winners  in $P$. Clearly  $|\cwinner(P)|\le 1$. %and the domain of $\cwinner(\cdot)$  can be naturally extended to all weighted or unweighted directed graphs.

%For example, a distribution $\hat\pi$, $\wmg(\hat\pi)$, and $\umg(\hat\pi)$ for $m=3$ are illustrated in Figure~\ref{fig:ex-m3}. %Only positive edges are shown in the WMG.  
%We have $\cwinner(\hat\pi) = \emptyset$.
%\begin{figure}[htp]
%\noindent\resizebox{1\textwidth}{!}{
%$\hat\pi = \left\{\begin{array}{ll}1\succ 2\succ 3 &\text{w.p. }1/4\\
%2\succ 1\succ 3 &\text{w.p. }1/4\\
%\text{other ranking} &\text{w.p. }1/8\\\end{array}\right.\Longrightarrow \wmg(\hat\pi) = $
%\begin{minipage}{0.15\linewidth}
%\includegraphics[width = \linewidth]{fig/ex-WMG-m3.pdf}
%\end{minipage}
%$\Longrightarrow \umg(\hat\pi) =$
%\begin{minipage}{0.15\linewidth}
%\includegraphics[width = \linewidth]{fig/ex-UMG-m3.pdf}
%\end{minipage}
%}
%\caption{\small  $\hat\pi$, $\wmg(\hat\pi)$ (only positive edges are shown), and $\umg(\hat\pi)$.\label{fig:ex-m3}}
%\end{figure}
%As another example, let $\piuni$ denote the uniform distribution over $\ml(\ma)$. Then, the weight on every edge in $\wmg(\piuni)$ is $0$ and $\umg(\piuni)$ does not contain any edge.

\myparagraph{\bf Condorcet criterion.}  All axioms studied in this paper are {\em per-profile axioms}~\citep{Xia2020:The-Smoothed}, each of which is modeled as  a function $\sat{X}$ that maps a voting rule $r$, a profile $P$, and a group size $B\ge 1$ to $\{0,1\}$, where $0$ (respectively $1$)  means that $r$ violates (respectively, satisfies) the axiom at  $P$ w.r.t.~group size $B$. Then, the classical (worst-case) satisfaction of the axiom under $r$ becomes $\min_{P} \sat{X}(r,P)$. 

For any voting rule, any profile $P$, and any $B\ge 1$,  {\em Condorcet criterion} is modeled as a function $\CC$ such that $\sat{\CC}(r,P,B)=1$ if and only if either (1) there is no Condorcet winner under $P$, or (2) the Condorcet winner is the winner of $P$ under $r$. Notice that technically $\CC$ does not depend on $B$, which is included for notational consistency. 

\myparagraph{Participation}  is modeled by a function $\Par$ such that 
$\sat{\Par}(r,P,B)=1$ if and only if no group of at most $B$ voters have incentive to abstain from voting; otherwise $\sat{\Par}(r,P,B)=0$. Formally, $\sat{\Par}(r,P,B)=0$ if and only if there exists a subvector $P'$ of $P$ such that for all $R\in P'$, $r(P')\succ_R r(P)$.  The simultaneous satisfaction of $\CC$ and $\Par$ is denoted by $\CnP$, such that $\sat{\CnP}(r,P,B)=1$ if and only if $\sat{\CC}(r,P,B)=1$ and $\sat{\Par}(r,P,B)=1$. 

%Due to the space constraint, we focus on presenting the result about $\CnP$ in the main text. The formal definitions of other axioms, i.e., {half-way monotonicity ($\HM$)}, Maskin monotonicity ($\MM$), and strategy-proofness ($\SP$), and the simultaneous satisfaction of $\CC$ and them, i.e., $\CnH$, $\CnM$, and $\CnS$, are deferred to Appendix~\ref{app:prelim}.

\myparagraph{Half-way monotonicity} is modeled by a function $\HM$ such that 
$\sat{\HM}(r,P,B)=1$ if and only if no coalition of at most $B$ voter have incentive to flip their votes so that the outcome is more preferred to each of them; otherwise $\sat{\Par}(r,P,B)=0$. 

\myparagraph{Maskin monotonicity} is modeled by a function $\MM$ such that $\sat{\MM}(r,P,B)=1$ if and only if whenever no more than $B$ voters raise the position of the current winner relative to other alternatives, the winner stays the same;  otherwise $\sat{\MM}(r,P,B)=0$. More precisely, given two rankings $R_1,R_2\in\ml(\ma)$ and an alternative $a$, we say that the position of $a$ is raised in $R_2$ relative to other alternatives from $R_1$ if and only if the set of alternatives $a$ is preferred to in $R_1$ is a subset of the set of alternatives $a$ is preferred to in $R_2$, i.e.,
$$\{b:a\succ_{R_1} b\}\subseteq \{b:a\succ_{R_2} b\}$$
Notice that the inclusion does not need to be strict. That is, if the the set of alternatives $a$ is preferred to is unchanged, then we still say that the position of $a$ is raised in $R_2$ relative to other alternatives.

\myparagraph{Strategy-proofness} is modeled by a function $\SP$ such that $\sat{\SP}(r,P,B)=1$ if and only if no coalition of at most $B$ voter have incentive to change their votes so that the outcome is more preferred to each of them; otherwise $\sat{\SP}(r,P,B)=0$. 

\myparagraph{Group versions of the axioms.} For any $X\in \{\Par,\HM,\MM,\SP\}$, if $\sat{X}(r,P,B)=1$, then we also say that $P$ satisfies $X_B$.

%In this paper we study the simultaneous satisfaction of $\CC$ and the each of the four axioms defined above. For example, the simultaneous satisfaction of $\CC$ and $\Par$ is denoted by $\CnP$, such that $\sat{\CnP}(r,P,B)=1$ if and only if $\sat{\CC}(r,P,B)=1$ and $\sat{\Par}(r,P,B)=1$. 

\vspace{2mm}\noindent{\bf Proof diagram of the $\CnP$ impossibility.} We briefly recall the proof diagram by~\citet[Chapter 1]{Peters2019:Fair}  in Figure~\ref{fig:proof} to show that when $m=4$, no voting rule $r$ satisfies $\CnP$, which will play an important role in our proofs later. %Consider the profile $P_0$ at the root.  
Each edge represents a sequence of operations, conditioned on the winner of the source profile being highlighted. \citet[Chapter 1]{Peters2019:Fair} proved that a violation of $\CC$ and/or $\Par$ exists in the diagram. Take the leftmost branch for example. If $r(P_0)\in \{1,2\}$, then two copies of $[1234]$ are added one by one. If the winner is no longer $1$ or $2$ during this process, then $\Par$ is violated. Otherwise, starting from $P_{\{1,2\}}$, if $r(P_{\{1,2\}}) = \{1\}$, then three votes of $[2431]$ are subtracted one by one. If the winner is not $1$ at any point, then $\Par$ is violated. However,$3$ is the Condorcet winner  in the leaf node, which means that $\CC$ is violated if $\Par$ has not been violated on the leftmost branch so far.

\begin{figure}[htp]
\centering  
\includegraphics[width = 1\textwidth]{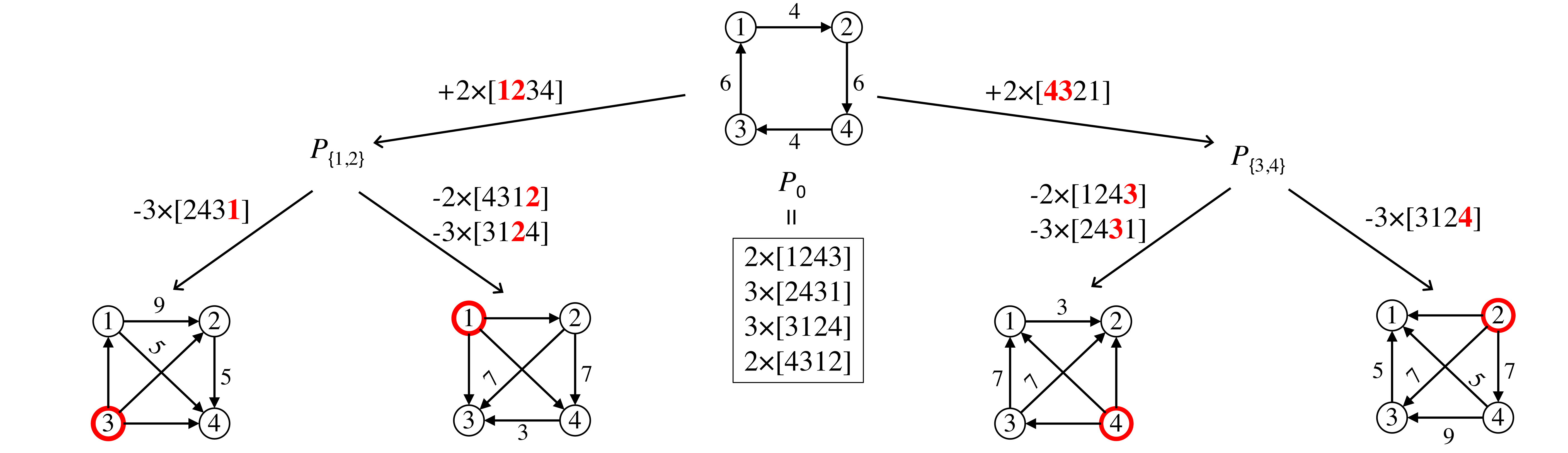}
\caption{Proof diagram of $\CnP$ impossibility~\citep[Chapter 1]{Peters2019:Fair}. WMGs of the root and the leaves are shown (where all unweighted edges in the leaves have weight $1$). Each edge represents a sequence of operations conditioned on the winner of the source profile being a highlighted alternative.  Condorcet winners in the leaf nodes are highlighted.\label{fig:proof}}
\end{figure}

\vspace{2mm}\noindent{\bf Semi-random satisfaction of axioms.}   Given a per-profile axiom $X$, a set $\Pi$ of distributions over $\ml(\ma)$, a voting rule $r$, $n\in\mathbb N$, and a group size $B$, the {\em semi-random satisfaction of $X$} under $r$ with $n$ agents, denoted by $\satmin{X}{\Pi}(r,n, B)$, is defined in Equation~(\ref{dfn:s-sat}) in the Introduction. The ``min'' in the superscript means that the adversary aims at minimizing the satisfaction of $X$. 

The semi-random analysis generalizes the classical quantitative analysis in social choice (under IC). To see this, let $\piuni$ denote the uniform distribution over $\ml(\ma)$ and let  $\Pi_{\text{IC}}=\{\piuni\}$. Then,  $\satmin{X}{\Pi_{\text{IC}}}$ becomes the likelihood of satisfaction of $X$ under IC. 
%For any $q\in\mathbb N$, a set $\Pi$ of distributions over $[q]$ is {\em strictly positive}, if there exists $\epsilon>0$ such that for every $\pi\in\Pi$ and every $j\in[q]$, $\pi(j)\ge \epsilon$. $\Pi$ is  {\em closed}, if it is a closed subset of the probability simplex in $\mathbb R^{q}$.  Let $\piuni$ is the uniform distribution 
Throughout the paper, we make the following  assumptions on $\Pi$.

\begin{asmp}%[\bf Strictly positive and closed $\Pi$] 
\label{asmp:strict} We assume that $\Pi$ is 
\begin{itemize}
\item {\em strictly positive}, which means that there exists $\epsilon>0$ such that for every $\pi\in\Pi$ and every $R\in\ml(\ma)$, $\pi(R)\ge \epsilon$;
\item {\em closed}, which means that $\Pi$ is a closed subset of the probability simplex in $\mathbb R^{m!}$; and 
\item $\piuni\in\conv(\Pi)$, where $\conv(\Pi)$ is the convex hull of $\Pi$.
\end{itemize}
%strictly positive and closed, and $\piuni\in\conv(\Pi)$.
\end{asmp}
The first part of Assumption~\ref{asmp:strict} requires that no distribution in $\Pi$ is too ``deterministic''. The second part is a mild technical assumption. The first two parts guarantee that the semi-random analysis using $\Pi$ is sufficiently different from the worst-case analysis.  The third part requires that the uniform distribution $\piuni$ is in the convex hull of $\Pi$, though $\piuni$ itself may not be in $\Pi$. 

%For example, let $\Pi_{\text{IC}}=\{\piuni\}$. Then, $\Pi_{\text{IC}}$ satisfies Assumption~\ref{asmp:strict}. Moreover, $\satmin{X}{\Pi_{\text{IC}}}(r,n)$ becomes the likelihood of satisfaction of $X$  under IC. 
%Let us look at another more informative example that illustrates smoothed analysis beyond i.i.d.~distributions.
We believe that Assumption~\ref{asmp:strict} is mild, because it is satisfied by many classical models for preferences. For example, it is satisfied by IC, which corresponds to $\Pi=\{\piuni\}$. The following example is taken from~\cite{Xia2020:The-Smoothed}. % to show that other models also satisfy  Assumption~\ref{asmp:strict}. 

\begin{ex}
\label{ex:Pi} In the {\em single-agent Mallows with bounded dispersion}, given $\underline\varphi>0$, each distribution is parameterized by a central ranking $W\in\ml(\ma)$ and a dispersion $\varphi\in [\underline\varphi,1]$, such that the probability for  $R\in \ml(\ma)$ is proportional to $\varphi^{\KT(R,W)}$, where $\KT(R,W)$ is the total number of  pairwise differences between $R$ and $W$, i.e., the {\em Kendall-Tau distance}.

In the {\em single-agent Plackett-Luce with bounded parameters}, given $\underline\varphi>0$, each distribution is parameterized by a vector $\vec \theta\in [\underline\varphi,1]^m$ such that $\vec \theta\cdot \vec 1 = 1$. The probability  for  $R=a_1\succ a_2\succ\cdots\succ a_m$ is $\prod_{i=1}^{m-1}({\theta_{a_i}}/{\sum_{\ell=i}^m \theta_{a_\ell}})$. 

%Let $m=3$ and $\Pi = \{\pi^1,\pi^2\}$, where $\pi^1$ and $\pi^2$ represents the following two distributions.
%\begin{center}
%\begin{tabular}{@{}c|c|c|c|c|c|c|}
%\cline{2-7}
%& $1\succ 2\succ 3$& $1\succ 3\succ 2$& $2\succ 1\succ 3$& $2\succ 3\succ 1$& $3\succ 1\succ 2$& $ 3\succ 2\succ 1$\\ 
%\hline 
%\begin{tabular}{|c}$\pi^1$\end{tabular}& $1/4$ & $1/12$& $1/6$& $1/6$& $1/6$& $1/6$\\
%\hline
%\begin{tabular}{|c}$\pi^2$\end{tabular}& $1/12$ & $1/4$& $1/6$& $1/6$& $1/6$& $1/6$\\
%\hline
%\end{tabular}
%\end{center}
%It follows that $\Pi$ is strictly positive (by $1/12$) and closed. Let $n=2$, we have
%\begin{align*}
%&\satmin{\nCP}{\Pi}(r,2) = \min\left\{\begin{array}{c}\Pr\nolimits_{P\sim(\pi^1,\pi^1)}(\sat{\nCP}(r,P) =1),\Pr\nolimits_{P\sim(\pi^1,\pi^2)}(\sat{\nCP}(r,P) =1),\\  \Pr\nolimits_{P\sim(\pi^2,\pi^1)}(\sat{\nCP}(r,P) =1),\Pr\nolimits_{P\sim(\pi^2,\pi^2)}(\sat{\nCP}(r,P) =1)\end{array}\right\}\end{align*}
For any $ \underline{\varphi}>0$, both models in the example satisfy Assumption~\ref{asmp:strict}. 
\end{ex}
If  $ \underline{\varphi}=0$ is allowed in Example~\ref{ex:Pi}, then the semi-random analysis degenerates to worst-case analysis, which trivializes the question.

 \section{Semi-Random Impossibility of  $\CC$ and $\Par$}
\label{sec:impossibility}
 
%The main result of this paper is the following smoothed impossibility theorem on the simultaneous satisfaction of $\CC$ and $\Par$. Let $\conv(\Pi)$ denote the {\em convex hull} of $\Pi$.

\begin{thm}[\bf \boldmath $\CC$+Participation]
\label{thm:CC+Par}
For any fixed $m\ge 4$, any  $\Pi$ that satisfies Assumption~\ref{asmp:strict}, any voting rule $r$, any $n\ge 12$,  and any $1\le B\le \sqrt n$,  
$$\satmin{\CnP}{\Pi}(r,n, B) = 1- \Omega\left(\frac{B}{\sqrt n}\right)$$
\end{thm} 
The theorem is more general than its classical, worst-case counterpart, as the likelihood is strictly smaller than $1$. It is also more general than its quantitative counterparts (under IC), because the latter is a special case of the former, where $\Pi=\{\piuni\}$, as discussed in the last section. 

\begin{proof} To illustrate the idea, we make the following assumptions in the proof sketch: (i) $m=4$, (ii) $\sqrt n$ is an integer, (iii) $B\mid \sqrt n$, and (iv) $m!\mid n$. Then, we modify the proof for the general case.  It suffices to prove the theorem when $n\ge 12$ and is sufficiently large, because the (worst-case) impossibility theorem holds for every $n\ge 12$~\citep{Brandt2017:Optimal,Peters2019:Fair}. 

\myparagraph{Overview.}  Let $\Par_B$ denote the group version of participation with size $B$. Instead of upper-bounding $\satmin{\CnP}{\Pi}$, we will lower-bound its complement $\satmax{\neg(\CnP)}{\Pi}$ as $\Omega(\frac{B}{\sqrt n})$, which is the max-semi-random likelihood for $\CC$ or $\Par_B$ to be violated and is defined similarly to $\satmin{X}{\Pi}$ in~\eqref{dfn:s-sat}, except that $\inf$ is replaced by $\sup$. The theorem  then follows after noticing  
%We will prove that the max-semi-random likelihood that $\CC$ or $\Par_B$ is violated (denoted by $\neg(\CnP)$) is $\Omega(\frac{B}{\sqrt n})$. Then, the theorem follows after noticing that 
$$\satmin{\CnP}{\Pi}(r,n, B) = 1 - \satmax{\neg(\CnP)}{\Pi}(r,n, B)$$
%where $\satmax{X}{\Pi}$ is defined similarly to $\satmin{X}{\Pi}$ in~\eqref{dfn:s-sat}, except that $\inf$ is replaced by $\sup$. 

%The main challenge in step 1 is to identify sufficiently many ``violation profiles'' in a neighborhood of an ``expected profile'' of some $\vec\pi\in\Pi^n$, yet we do not have much information about the violation profiles, as the voting rule $r$ can be arbitrary. This is addressed by creating {\em violation templates} that scale up existing proof diagrams, e.g.,~\citep[Chapter 1]{Peters2019:Fair} by a factor of $\sqrt n$, applying the violation templates to certain profiles, and then using a flipped diagram to prove that there are sufficiently many violation profiles as desired. The main challenge in step 2 is that it is unclear how such $\vec\pi$ can be chosen, because the violation profiles identified in step 1 can be quite arbitrary. This is addressed by using an averaging argument over all $n!$ permutations of a carefully chosen $\vec\pi$, converting the likelihood to the likelihood about the histogram of profiles, and then applying the point-wise concentration bound~\citep[Lemma~1]{Xia2021:How-Likely} to provide the lower bound on the likelihood of violation. 

As discussed in Section~\ref{sec:contributions},  the proof proceeds in two steps. In step 1, we identify a set of $n$-profiles, denoted by $\calV_{n,B}$, where $\CC$ or $\Par_B$ is violated, and prove that $\calV_{n,B}$ contains sufficiently many profiles. This will be achieved by first scaling the (worst-case) proof diagram in~\citep[Chapter 1]{Peters2019:Fair}, i.e., Figure~\ref{fig:proof} in Section~\ref{sec:prelim}, by a factor of $\sqrt n$ to define a {\em violation template}, and then implementing it at profiles whose histograms are in an $O(\sqrt n)$ neighborhood of $\frac{n}{m!}\cdot \vec 1$. Each implementation leads to a {\em violation tree}, which contains at least one violation of $\CC$ or $\Par_B$. Then, we upper-bound the number of violation trees any  profile $P^*\in\calV_{n,B}$ can be on, by considering the {\em rotated trees} generated by the {\em rotated template} rooted at $P^*$.

Then in step 2, we prove that there exists $\vec\pi\in\Pi^n$ so that the likelihood of $\calV_{n,B}$ found in step 1 is lower-bounded by $\Omega(\frac{B}{\sqrt n})$. This is achieved by starting with a $\vec\pi\in\Pi^n$ such that $\sum_{j=1}^n \pi_j$ is $O(1)$ away from $\frac{n}{m!}\cdot \vec 1$, and then considering the sum of likelihood of $\calV_{n,B}$ under all $n!$ permutations of components in $\vec \pi$. This converts the likelihood of $\calV_{n,B}$ to the likelihood about the histogram of a randomly-generated profile. Finally, we apply the point-wise concentration bound~\citep[Lemma~1]{Xia2021:How-Likely} to derive the desired lower bound. 

\myparagraph{Step 1.} We first formally define the violation template illustrated in Figure~\ref{fig:violation-template}.

\begin{figure}[htp]
\centering  
\includegraphics[width = \textwidth]{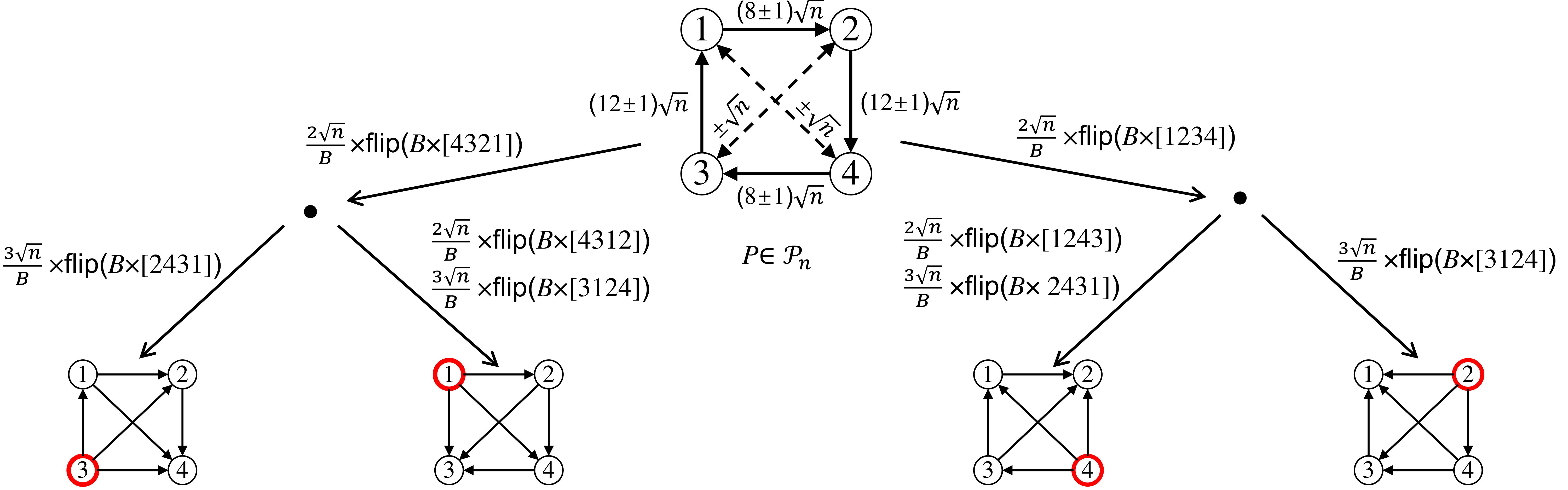}
\caption{The violation template. ``$\pm$'' represents the  upper and lower bounds on the weight. $\ell\times \text{\rm flip}(B\times R)$ represents $\ell$ operations, each of which flips $B$ votes that are $R$. In the leaves, weights are not shown and Condorcet winners  are highlighted. \label{fig:violation-template}}
\end{figure}

%\begin{figure}[htp]
%\centering  
%\includegraphics[width = \textwidth]{fig/trials.pdf}
%\caption{A trial and its augmented trial.}
%\end{figure}

\begin{dfn}[{\bf Violation template}{}]
\label{dfn:vio-diagram}
Given any $n$-profile $P$ with at least $7\sqrt n $ copies of $\ml(\ma)$ and any $1\le B\le \sqrt n$, a {\em violation template}  is defined by modifying the proof diagram (Figure~\ref{fig:proof}) as follows, where $\rev{R}$ denote the reverse ranking of $R$:
\begin{itemize}
\item every $+R$ operation on an edge in Figure~\ref{fig:proof} is replaced by a sequence of $\frac{\sqrt n}{B}$ operations, each of which flips $B\times \rev{R}$ votes and is denoted by $\flip{B\times \rev{R}}$;
\item every $-R$ operation on an edge in Figure~\ref{fig:proof} is replaced by a sequence of $\frac{\sqrt n}{B}$ operations, each of which flips $B \times  R$ votes  and is denoted by $\flip{B\times R}$.
\end{itemize}
\end{dfn}
%Given any profile $P$ with at least $3\sqrt n$ copies of $\ml(\ma)$, we can implement the violation template by letting the root be $P$.
The violation template will be implemented multiple times, by letting its root to be  $n$-profiles whose WMGs are similar to the WMG at the root in Figure~\ref{fig:proof} (scaled by a factor of $\sqrt n$) and whose histograms are close to $\frac{n}{m!}\cdot \vec 1$. Formally, we define the set of such profiles $P$, denoted by $\calP_n$, as follows. Let $w(e)$ denote the weight on 
edge $e$ in the root of Figure~\ref{fig:proof}.

\begin{dfn}\label{dfn:Pn}
Let  $\calP_n$ denote the set of $n$-profiles $P$ such that
\begin{itemize}
\item for every edge $e\in [4]\times [4]$, $|w_P(e) - \sqrt n \cdot w(e)|\le \sqrt n$;
\item  $|\hist(P)-\frac{n}{m!}\cdot\vec 1|_\infty\le 4 \sqrt n$.
\end{itemize}
\end{dfn}

%Figure~\ref{fig:violation-template} illustrates the violation template rooted at a profile $P$ in $\calP_n$. 
We note that any flip operation in the violation template  does not completely specify the set of voters whose votes should be flipped (except that their votes must be the same as the designated ranking). Consequently, different combinations of votes when implementing  the violation template   lead to different {\em violation trees}, formally defined as follows.
\begin{dfn}[{\bf Violation trees}{}]
\label{dfn:vio-tree}
For any $P\in\calP_n$  and any $1\le B\le \sqrt n$, a {\em violation tree} is a tree of $\frac{20\sqrt n}{B}+1$ profiles obtained from implementing the violation template rooted at $P$. Let $\calT_{P,B}$ denote the set of all violation trees rooted at $P$.
\end{dfn}
For example, in a violation tree rooted at $P\in\calP_n$, the  left branch of Figure~\ref{fig:violation-template} contains $\frac{5\sqrt n}{B}+1$ profiles. The first profile is $P$, and each of the following $\frac{2\sqrt n}{B}$ profiles are obtained from the previous one by flipping $B$ votes of $[4321]$. Then, the branch continues with $\frac{3\sqrt n}{B}$ profiles, each of which is obtained from the previous profile by flipping $B$ votes of $[2431]$.  The following claim lower-bounds the size of $\calT_{P,B}$. %All missing proofs can be found in the appendix.
\begin{claimnew}
\label{claim:vio-tree-num}
For any $P\in \calP_n$ and any $1\le B\le \sqrt n$,  $|\calT_{P,B}| = \Omega\left(\left(\frac{n}{m!  \sqrt[B]{B!}}\right)^{20\sqrt n}\right)$.
\end{claimnew}
\begin{proof}
For any profile $P'$ on any violation tree in $\calT_{P,B}$, $\hist(P')$ is in the $7\sqrt n$ neighborhood of $\frac{n}{m!}\cdot\vec 1$ in $L_\infty$. For any operation $\text{flip}(B\times R)$, let $n_R$ denote the number of $R$ votes in $P'$. Then, the number of combinations of $B$ votes, each of which being $R$, is ${n_R\choose B}$, which is at least  $\frac{( \frac{n}{m!}-8\sqrt n )^B}{B! }$. % Therefore, every operation in the implementation of the violation template, which involves flipping $B$ votes of the same ranking, has no more than $\frac{( \frac{n}{m!}-7\sqrt n )^B}{B! }$ combinations. 
Because the total number of operations in a violation template is $\frac{20\sqrt n}{B}$, for any sufficiently large $n$, the total number of violation trees rooted at $P'$ is at least
\begin{align*}
&\left(\frac{ (\frac{n}{m!}-8\sqrt n)^B }{B! }\right)^{\frac{20\sqrt n}{B}} \\
= &\left(\frac{ n}{m!\sqrt[B]{B!} }\right)^{ 20\sqrt n }
\cdot \left(1-\frac{8m!}{\sqrt n}\right)^{ 20\sqrt n } \\
=&\Omega\left(\left(\frac{n}{m!  \sqrt[B]{B!}}\right)^{20\sqrt n}\right) 
\end{align*}
\end{proof}

%\begin{proof}
%For any profile $P'$ on any violation tree in $\calT_{P,B}$, $\hist(P')$ is in the $7\sqrt n$ neighborhood of $\frac{n}{m!}\cdot\vec 1$ in $L_\infty$. For any operation $\text{flip}(B\times R)$, let $n_R$ denote the number of $R$ votes in $P'$. Then, the number of combinations of $B$ votes, each of which is $R$, is ${n_R\choose B}$, which is at least  $\frac{( \frac{n}{m!}-8\sqrt n )^B}{B! }$. % Therefore, every operation in the implementation of the violation template, which involves flipping $B$ votes of the same ranking, has no more than $\frac{( \frac{n}{m!}-7\sqrt n )^B}{B! }$ combinations. 
%Because the total number of operations in a violation template is $\frac{20\sqrt n}{B}$, for any sufficiently large $n$, the total number of violation trees rooted at $P'$ is at least
%\begin{align*}
%&\left(\frac{ (\frac{n}{m!}-8\sqrt n)^B }{B! }\right)^{\frac{20\sqrt n}{B}} \\
%= &\left(\frac{ n}{m!\sqrt[B]{B!} }\right)^{ 20\sqrt n }
%\cdot \left(1-\frac{8m!}{\sqrt n}\right)^{ 20\sqrt n } \\
%=&\Omega\left(\left(\frac{n}{m!  \sqrt[B]{B!}}\right)^{20\sqrt n}\right) 
%\end{align*}
%\end{proof}

The next claim states that each violation tree contains a violation of $\CC$ or $\Par_B$.
\begin{claimnew}
\label{claim:tree-violation}
For every $P\in \calP_{n}$ and every  $T\in \calT_{P,B}$,  $\CC$ or $\Par_B$ is violated in $T$.
\end{claimnew}
\begin{proof}
Following the proof diagram by~\citep[Chapter 1]{Peters2019:Fair} (also see Section~\ref{sec:prelim} and Figure~\ref{fig:proof}), there exists a pair of profiles $P_1, P_2$ in $T$ such that $P_2$ is obtained from $P_1$ by an operation, which flips $B$ votes of $R$, and either $P_2$ violates $\CC$ or $r(P_2)\succ_R r(P_1)$. If  $P_2$ violates $\CC$ then the claim automatically holds. Otherwise, we have $r(P_2)\succ_R r(P_1)$. Let $\{n_1,n_2,\ldots,n_B\}\subseteq [n]$ denote the indices to the flipped votes in $P_1$ in order to reach $P_2$.  For every $1\le t\le B$, let $P_1^t$ (respectively, $P_2^t$) denote the $(n-t)$-profile that is obtained from $P_1$ (respectively, $P_2$) by removing voters  $\{n_1,\ldots,n_t\}$. It follows that $P_1^B = P_2^B$. Then, there are two cases.
\begin{itemize}
\item If $r(P_1^B)\succ_R r(P_1)$, then  $\Par_B$ is violated at $P_1$, because the $B$ voters $\{n_1,\ldots,n_B\}$ can improve the winner by abstaining from voting. 
\item Otherwise, because  $r(P_2)\succ_R r(P_1)$, we have $r(P_2)\succ_R r(P_1^B) =r(P_2^B)$, which means that $r(P_2^B)\succ_{\rev{R}} r(P_2)$. Therefore, $\Par_B$ is violated at $P_2$, because the $B$ voters $\{n_1,\ldots,n_B\}$, whose votes are $\rev{R}$, can improve the winner by abstaining from voting.
\end{itemize}
In both cases $\CC$ and/or $\Par_B$ is violated at a profile in $T$, which proves Claim~\ref{claim:tree-violation}.
\end{proof}

Let $\calV_{n,B}$ denote the set of all profiles on violation trees $\bigcup_{P\in\calP_{n}} \calT_{P,B}$, where $\CC$  or $\Par_B$ is violated. The next claim upper-bounds the number of violation trees that each profile in $\calV_{n,B}$ can possibly be on.
%\end{dfn}

\begin{claimnew}
\label{claim:num-rotated-trees}
Every $P^*\in \calV_{n,B}$ is on no more than $O\left(\frac{\sqrt n}{B}\cdot\left(\frac{n}{m!  \sqrt[B]{B!}}\right)^{20\sqrt n}\right)$  violation trees rooted in $\calP_n$.
\end{claimnew}
\begin{proof}
For every node $V$ in the violation template, we define a {\em rotated template}, which reverses all edges along the path from the root to $V$. That is, an edge $V_1\ra V_2$ along the path that flips $B\times R$ becomes $V_1\leftarrow V_2$  in the rotated template that flips $B\times \rev{R}$. Consequently, the rotated template is a diagram rooted at $V$. Because each violation template has $\frac{20\sqrt n}{B}+1$ nodes, there are $\frac{20\sqrt n}{B}+1$ rotated templates. Figure~\ref{fig:rotated} illustrates a rotated template rooted at a node in the leftmost branch. 

\begin{figure}[htp]
\centering  
\includegraphics[width = .9\textwidth]{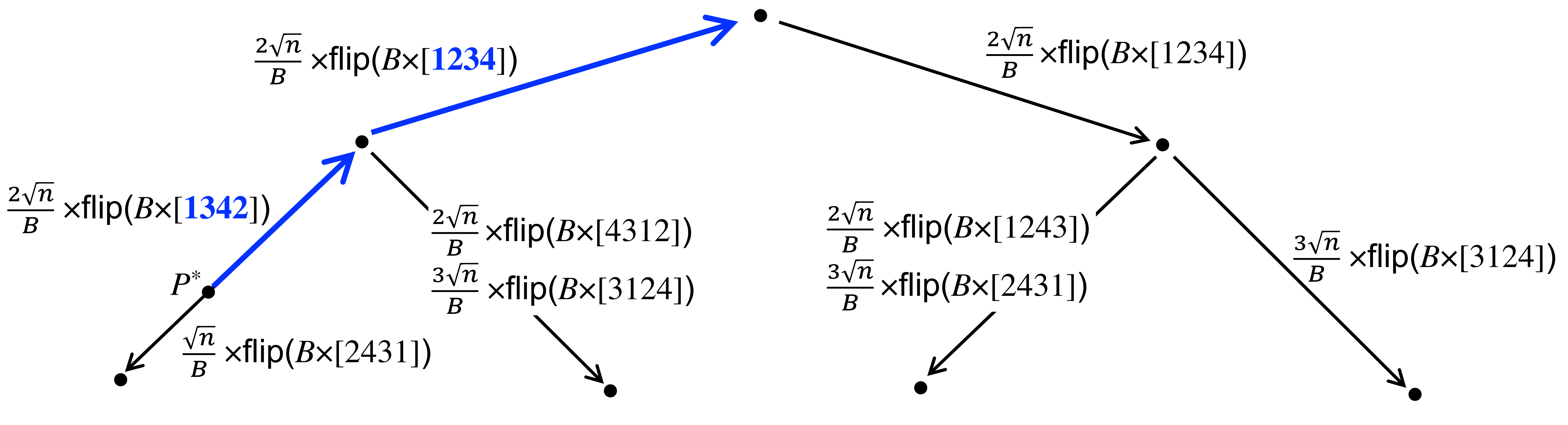}
\caption{A rotated template rooted at $P^*$. Reversed edges and rankings are highlighted.\label{fig:rotated}}
\end{figure}

Like the violation templates,  any rotated template leads to {rotated trees} rooted at any profile $P^*$ that contains sufficiently many copies of $\ml(\ma)$. Given a rotated template and  $P^*\in \calV_{n,B}$, because violation trees are rooted at profiles in $P_n$, and each node in a violation tree is obtained from the root by modifying no more than $7\sqrt n$ votes, we have  $|\hist(P^*) - \frac{n}{m!}\cdot\vec 1| \le 11 \sqrt n$. This means that the histogram of any profile on the rotated tree is no more than $18\sqrt n$ away from $\frac{n}{m!} \cdot\vec 1$. Therefore, like the proof of Claim~\ref{claim:vio-tree-num}, the number of rotated trees rooted at $P^*$ is at most
\begin{align*}
&\left(\frac{ (\frac{n}{m!}+18\sqrt n)^B }{B! }\right)^{\frac{20\sqrt n}{B}} \\
=& \left(\frac{ n}{m!\sqrt[B]{B!} }\right)^{ 20\sqrt n }
\cdot \left(1+ \frac{18m!}{\sqrt n}\right)^{ 20\sqrt n }\\
 =&O\left(\left(\frac{n}{m!  \sqrt[B]{B!}}\right)^{20\sqrt n}\right) 
\end{align*}
Because there are $\frac{20\sqrt n}{B}+1$ rotated templates and each violation tree that contains $P^*$ is equivalent to a rotated tree rooted at $P^*$, the number of violation trees that contains $P^*$ is  no more than 
\begin{align*}
&\left(\frac{20\sqrt n}{B}+1\right)\cdot O\left(\left(\frac{n}{m!  \sqrt[B]{B!}}\right)^{20\sqrt n}\right)  \\
=& O\left(\frac{\sqrt n}{B}\cdot \left(\frac{n}{m!  \sqrt[B]{B!}}\right)^{20\sqrt n}\right)
\end{align*}
This completes the proof of Claim~\ref{claim:num-rotated-trees}.
\end{proof}

We are now ready to lower-bound $|\calV_{n,B}|$. To this end, we count the number of (profile, violation tree) pairs, denoted by $(P,T)$, where $T$ is rooted at a profile in $\calP_n$, $P$ is on $T$, and $\CC$ and/or $\Par_B$ is violated at $P$. By Claim~\ref{claim:vio-tree-num} and Claim~\ref{claim:tree-violation}, the total number of such (profile, violation tree) pairs is at least $|\calP_n|\times \Omega\left(\left(\frac{n}{m!  \sqrt[B]{B!}}\right)^{20\sqrt n}\right)$. By Claim~\ref{claim:num-rotated-trees}, the total number of such (profile, violation tree) pairs is at most $|\calV_{n,B}|\times O\left(\frac{\sqrt n}{B}\cdot \left(\frac{n}{m!  \sqrt[B]{B!}}\right)^{20\sqrt n}\right)$. Therefore,
\begin{equation}
\label{eq:vpratio}
\frac{|\calV_{n,B}|}{|\calP_n|}\ge \frac{\Omega\left(\left(\frac{n}{m!  \sqrt[B]{B!}}\right)^{20\sqrt n}\right)}{O\left(\frac{\sqrt n}{B}\cdot \left(\frac{n}{m!  \sqrt[B]{B!}}\right)^{20\sqrt n}\right)}=\Omega\left(\frac{B}{\sqrt n}\right)
\end{equation}

\myparagraph{Step 2.} %We first choose $\vec \pi^*\in\Pi^n$ and prove that a permutation of it achieves the lower bound.  %\eqref{eq:vpratio} immediately implies that the theorem holds under IC, because each profile appears equally likely. 
Let $\vec\pi\in\Pi^n$ be such that $\sum_{j=1}^n \pi_j$ is $O(1)$ away from $\frac{n}{m!}\cdot\vec 1$, which can be defined by rounding as shown in~\citep{Xia2021:How-Likely}. Let $\calS_n$ denote the set of all permutations over $[n]$. For any permutation $\eta\in\calS_n$, let $\eta(\vec\pi)$ denote the vector of distributions where the indices are permuted according to $\eta$. That is, $\eta(\vec\pi) = (\pi_{\eta(1)},\ldots,\pi_{\eta(n)})$. We prove that there exists a permutation $\eta$ over $[n]$, such that 
\begin{equation}
\label{eq:vb}
\Pr\nolimits_{P\sim \eta(\vec\pi)}(P\in \calV_{n,B}) = \Omega\left(\frac{B}{\sqrt n}\right)
\end{equation}
It suffices to prove that the sum of  the left hand side of \eqref{eq:vb} for all $\eta\in\calS_n$ is at least $n!$ times the right hand side of \eqref{eq:vb}, that is,
\begin{equation}
\label{eq:Vnb}
\sum\nolimits_{\eta\in \calS_n}\Pr\nolimits_{P\sim \eta(\vec\pi)}(P\in \calV_{n,B}) = \Omega\left(n!\cdot\frac{B}{\sqrt n}\right)
\end{equation}
Nevertheless, \eqref{eq:Vnb} is still hard to prove due to the lack of information about the profiles in $\calV_{n,B}$. The key insight of our proof is to convert the left hand side of \eqref{eq:Vnb} to probabilities for the histogram of $P$ to be the histograms of profiles in $\calV_{n,B}$. % each of which is $O(\sqrt n)$ away from $\frac{n}{m!}\cdot\vec 1$, and therefore the probabilities can be lower-bounded by the point-wise concentration bound~\citep[Lemma~1]{Xia2021:How-Likely}. 
Notice that the left hand side of  \eqref{eq:Vnb} is equivalent to 
%$$\sum\nolimits_{\eta\in \calS_n}\sum\nolimits_{P^*\in  \calV_{n,B}}\Pr\nolimits_{P\sim \eta(\vec\pi)}(P = P^*)$$
%For every $P^*\in \calV_{n,B}$, notice that 
%$$\Pr\nolimits_{P\sim \eta(\vec\pi)}(P = P^*) = \Pr\nolimits_{P\sim  \vec\pi }(P = \eta(P^*))$$
%The key insight is the following equation.
\begin{align}
&\sum\nolimits_{\eta\in \calS_n}\sum\nolimits_{P^*\in  \calV_{n,B}}\Pr\nolimits_{P\sim \eta(\vec\pi)}(P = P^*)\notag\\
=&\sum\nolimits_{P^*\in  \calV_{n,B}}\sum\nolimits_{\eta\in \calS_n}\Pr\nolimits_{P\sim  \vec\pi }(P = \eta^{-1}(P^*))\notag\\
=&\sum\nolimits_{P^*\in  \calV_{n,B}} \sum\nolimits_{\eta\in \calS_n}\Pr\nolimits_{P\sim \vec\pi}(P =\eta^{-1}(P^*))\notag\\
=& \sum\nolimits_{P^*\in  \calV_{n,B}}\prod\nolimits_{R\in\ml(\ma)}(n^*_R!)\times \notag\\
&\hspace{17mm}\Pr\nolimits_{P\sim \vec\pi}(\hist(P) =\hist(P^*)),\label{eq:hist-P}
\end{align}
where $n^*_R$ is the number of $R$ votes in $P^*$. \eqref{eq:hist-P} follows the following claim. %, whose proof can be found in  Appendix~\ref{app:hist-conversion}.
\begin{claimnew}
\label{claim:hist-conversion} For any $P^*\in  \calV_{n,B}$, 
%\resizebox{.5\textwidth}{!}{\parbox{.5\textwidth}{
\begin{align*}&\sum\nolimits_{\eta\in \calS_n}\Pr\nolimits_{P\sim \vec\pi}(P =\eta^{-1}(P^*))=   \prod\nolimits_{R\in\ml(\ma)}(n^*_R!)\times\Pr\nolimits_{P\sim \vec\pi}(\hist(P) =\hist(P^*))
\end{align*}
%}}
\end{claimnew}

\begin{proof}
Let $M$ (respectively, $\hat M$) denote the set (respectively, multi-set) of $n$-profiles $\{\eta^{-1}(P^*):\eta\in \calS_n\}$. It follows that $\hat M$ consists of $\prod\nolimits_{R\in\ml(\ma)}(n^*_R!)$ copies of $M$. Therefore, 
\begin{align*}
& \sum\nolimits_{\eta\in \calS_n}\Pr\nolimits_{P\sim \vec\pi}(P =\eta^{-1}(P^*))\notag\\
=& \sum \nolimits_{P'\in \hat M}\Pr\nolimits_{P\sim \vec\pi}(P = P')\\
=& \prod\nolimits_{R\in\ml(\ma)}(n^*_R!) \sum \nolimits_{P'\in M}\Pr\nolimits_{P\sim \vec\pi}(P = P')\\
=& \prod\nolimits_{R\in\ml(\ma)}(n^*_R!)\Pr\nolimits_{P\sim \vec\pi}(\hist(P) =\hist(P^*)),
\end{align*}
which proves the claim.
\end{proof}

%that the permutation on $\vec\pi$ in the subscribe  convert the left hand side of \eqref{eq:Vnb} to probabilities for the histogram of a randomly generated profile to be $\hist(P^*)$, by noticing that 
%$$\sum\nolimits_{\eta\in \calS_n}\Pr\nolimits_{P\sim \eta(\vec\pi)}(P = P^*) = 
%\sum\nolimits_{\eta\in \calS_n}\Pr\nolimits_{P\sim  \vec\pi }(P = \eta(P^*))$$
%\lirong{give an overview of the rest of the proof}Notice that
%\begin{align*}
%&\sum\nolimits_{\eta\in \calS_n}\Pr\nolimits_{P\sim \eta(\vec\pi)}(P\in \calV_{n,B})  \\
%=& \sum\nolimits_{\eta\in \calS_n}\sum\nolimits_{P^*\in  \calV_{n,B}}\Pr\nolimits_{P\sim \vec\pi}(P =\eta(P^*))\\
%= &\sum\nolimits_{P^*\in  \calV_{n,B}}\sum\nolimits_{\eta\in \calS_n}\Pr\nolimits_{P\sim \vec\pi}(P =\eta(P^*))\\
%=&\sum\nolimits_{P^*\in  \calV_{n,B}}\prod\nolimits_{R\in\ml(\ma)}(n_R!)\Pr\nolimits_{P\sim \vec\pi}(\hist(P) =\hist(P^*)),
%\end{align*}
%where $n_R = [\hist(P^*)]_{R}$ denote the number of $R$ votes in $P^*$. 

Because $\hist(P^*)$ is $7\sqrt n$ away  from $\frac{n}{m!}\cdot\vec 1$, we have $\prod\nolimits_{R\in\ml(\ma)}(n^*_R!) = \Theta\left(\left(\frac{n}{m!}\right)!\right)$. By the point-wise concentration lemma~\cite[Lemma~1]{Xia2021:How-Likely}, the likelihood for the histogram of $P$ to be any $\vec x$ in an $O(\sqrt n)$ neighborhood of $\sum_{j=1}^n\pi_j$ is $\Omega(n^{(1-m!)/2})$. Recall that $\sum_{j=1}^n\pi_j$ is $O(1)$ away from $\frac{n}{m!}\cdot\vec 1$. Therefore, for every $P^*\in\calV_{n,B}$, 
\begin{equation}
\label{eq:hist-lower}
\Pr\nolimits_{P\sim \vec\pi}(\hist(P) =\hist(P^*)) = \Omega\left(n^{(1-m!)/2}\right)
\end{equation}
Combining  \eqref{eq:vpratio}, \eqref{eq:hist-P}, and \eqref{eq:hist-lower},   the left hand side of \eqref{eq:Vnb} becomes
\begin{align}
%\label{eq:Vnb-size}
%&\sum\nolimits_{\eta\in \calS_n}\Pr\nolimits_{P\sim \eta(\vec\pi)}(P\in \calV_{n,B})\notag \\
& \Omega\left(|\calV_{n,B}|\cdot \left(\frac{n}{m!}\right)!\cdot n^{(1-m!)/2}\right)\notag\\
=& \Omega\left(\frac{B}{\sqrt n }\cdot \left(\frac{n}{m!}\right)!\cdot |\calP_{n}|\cdot n^{(1-m!)/2}\right) \label{eq:Vnb-size}
\end{align}
Next, we lower-bound $|\calP_{n}|$ in the following claim. %, whose proof is in Appendix~\ref{app:proof-Pn}.
\begin{claimnew}
\label{claim:Pn}
%\begin{equation}
%\label{eq:lb-Pn}
$|\calP_{n}| = \Omega\left(\frac{n!}{\left(\frac{n}{m!}\right)!}\cdot n^{(m!-1)/2}\right)$.
%\end{equation} 
\end{claimnew}
\begin{proof}
 Let $P$ denote the $n$-profile obtained from any profile whose histogram is $\frac{n}{m!}\cdot\vec 1$ by flipping $2\sqrt n\times [1234]$, $3\sqrt n\times [2431]$, $3\sqrt n\times [3124]$, and $2\sqrt n\times [4312]$. Then, we have  $P\in \calP_n$. Moreover, let $\calN$ denote the set of $n$-profiles whose histograms are no more than $\frac{\sqrt n}{m!}$ away from $\hist(P)$ in $L_\infty$. It follows that $\calN\subseteq \calP_n$ and the histogram of any profile in $\calP_n$ is no more than $\frac{4\sqrt n}{m!}$ away from $\frac{n}{m!}\cdot\vec 1$. Therefore,
\begin{align*}
|\calP_{n}| \ge |\calN|\ge & \left(\frac{ \sqrt n}{m!}\right)^{m!-1}\times  \frac{n!}{\left(\frac{n}{m!}+4\sqrt n\right)^{m!}} \\
= &\Omega\left(\left(\sqrt n\right)^{(m!-1)/2}\right)\times \Omega\left(\frac{n!}{\left(\frac{n}{m!}\right)^{m!}}\right),
\end{align*}
where $\left(\frac{ \sqrt n}{m!}\right)^{m!-1}$ is a lower bound on the number of histograms in the $\frac{\sqrt n}{m!}$ neighborhood of $\hist(P)$, and for every such histogram $\vec h$,  $\frac{n!}{\left(\frac{n}{m!}+4\sqrt n\right)^{m!}}$ is a lower bound on the number of $n$-profiles whose WMG is $\vec h$. This completes the proof of Claim~\ref{claim:Pn}.
\end{proof}

Finally, \eqref{eq:Vnb} follows after \eqref{eq:Vnb-size} and Claim~\ref{claim:Pn}, which completes the proof for the special case of Theorem~\ref{thm:CC+Par}. 

\myparagraph{The general case.} We make the following three changes.

\myparagraph{First,} in the definition of $\calP_n$ (Definition~\ref{dfn:Pn}), the weight on any edge from $\{1,2,3,4\}$  to $\{5,\ldots, m\}$ is required to be at least $20\sqrt n$. The new $\calP_n$ is illustrated in the root of Figure~\ref{fig:VD-general}.

\myparagraph{Second,} we  change  the violation template (Definition~\ref{dfn:vio-diagram}) as illustrated in Figure~\ref{fig:VD-general}.
\begin{itemize}
\item We add an extra branch from the root that performs $\frac{7\sqrt n}{B}\times \text{flip}(B\times [1\succ 2\succ 3\succ 4\succ 5\succ\cdots\succ m])$. $4$ is the Condorcet winner at the leaf node. We also append $5\succ \cdots\succ m$ to the end of every ranking to be flipped in the violation template.
\item Every edge in Figure~\ref{fig:proof} is replaced by a sequence of $\lceil\frac{\lceil\sqrt n\rceil}{B}\rceil$ operations: each of the first $\lceil\frac{\lceil\sqrt n\rceil}{B}\rceil-1$ operations flips $B$ rankings, and the last flips $\lceil\sqrt n\rceil - B\cdot\left(\lceil\frac{\lceil\sqrt n\rceil}{B}\rceil-1\right)$ rankings. 
\end{itemize}
\begin{figure}[htp]
\centering  
\includegraphics[width = \textwidth]{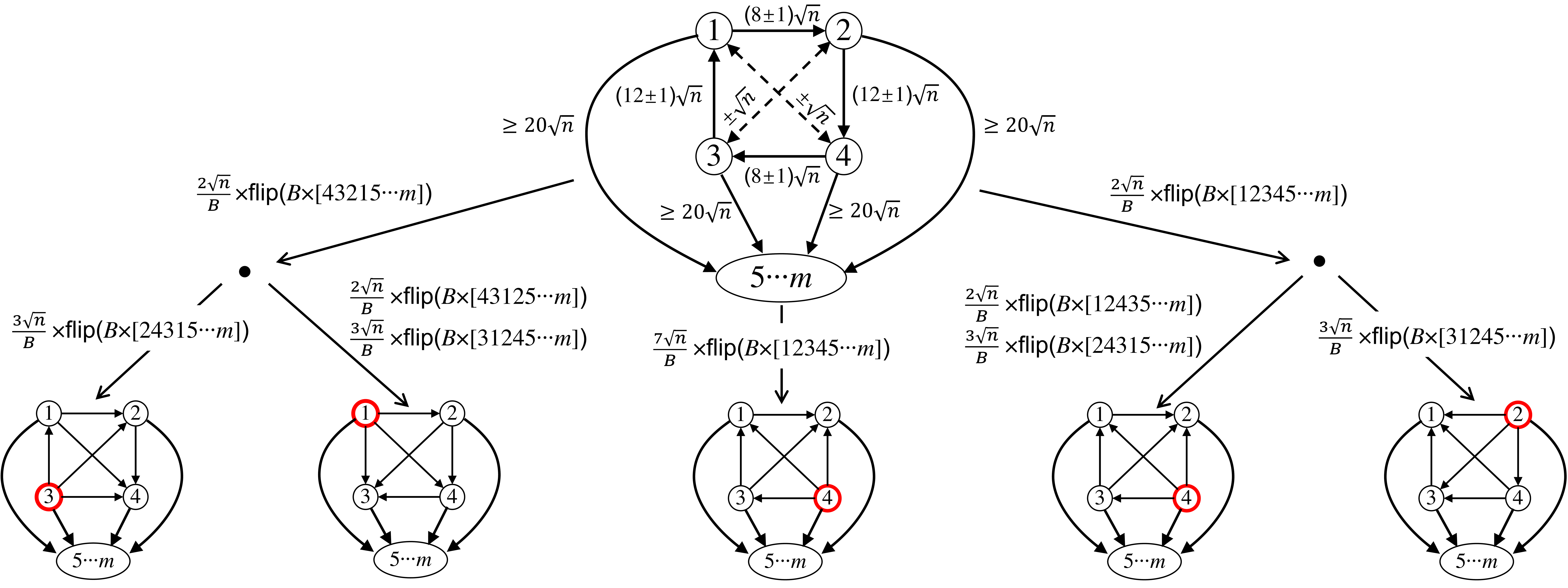}
\caption{Violation template for $\CnP$, general case.\label{fig:VD-general}}
\end{figure}
\myparagraph{Third,} $\frac{n}{m!}\cdot \vec 1$ is replaced by an arbitrary (but fixed) integer $\vec u = (u_1,\ldots,u_{m!})$ via rounding, such that $|\vec u - \frac{n}{m!}\cdot \vec 1|_\infty\le 1$; $\left(\frac{n}{m!}\right)!$ is replaced by $\prod_{i=1}^{m!}(u_i!)$; and $\frac{n}{m!}+\alpha\sqrt n$ for any $\alpha$ appeared in the proof above is replaced by $\left\lceil\frac{n}{m!}\right\rceil+\alpha\sqrt n$.
\end{proof}

\section{Other Semi-Random Impossibilities}
\label{sec:other-imp}
The proof of Theorem~\ref{thm:CC+Par} can leverage any proof diagram like Figure~\ref{fig:proof}, that has the following three high-level features:
\begin{enumerate}
\item The diagram consists of constantly many nodes.
\item The diagram works all slightly perturbed root profiles. 
%rooted at all profiles whose histogram is in an $\Omega(\sqrt n)$ neighborhood of a vector that is $O(\sqrt n)$ away from $\frac{n}{m!}\cdot\vec 1$.
\item $X$ is violated in each violation tree (scaled by $\sqrt n$) .
\end{enumerate}
Therefore,  any existing proof diagram for $\CnP$, e.g., the one used in~\citep{Brandt2017:Optimal}, can be used to prove Theorem~\ref{thm:CC+Par}. We chose the diagram in~\citep[Chapter 1]{Peters2019:Fair} for its simplicity.  

In this section, we prove semi-random impossibilities for the other three combinations of axioms, i.e., $\CnH$, $\CnM$, and $\CnS$, by leveraging existing proof diagrams that have the three aforementioned features. %Semi-random impossibility for $\CnS$ follows as a corollary of the semi-random impossibility for $\CnH$.  
\begin{thm}[\bf $\CC$+half-way monotonicity]
\label{thm:CC+HM}
For any fixed $m\ge 4$, any $\Pi$ that satisfies Assumption~\ref{asmp:strict},   any voting rule $r$, any $n\ge 24$,  and any $1\le B\le \sqrt n$,  
$$\satmin{\CnH}{\Pi}(r,n, B) = 1- \Omega\left(\frac{B}{\sqrt n}\right)$$
\end{thm}
%The proof  leverages the same diagram~\citep[Chapter 1]{Peters2019:Fair} as in the proof of Theorem. The only difference is to verify that it has the third feature above. The full proof can be found in Appendix~\ref{app:proof-CC+HM}. 

\begin{proof}
The proof is similar to the proof of Theorem~\ref{thm:CC+Par} and uses the same violation template (Figure~\ref{fig:violation-template}). The main difference is that Claim~\ref{claim:tree-violation} is replaced by the following claim.
\begin{claimnew}  For every $P\in \calP_{n}$ and every  $T\in \calT_{P,B}$,  $\CC$ or $\HM_B$ is violated in $T$.
\end{claimnew}
The proof of this claim  is straightforward: suppose there exists a pair of profiles $P_1, P_2$ on a violation tree, such that $P_2$ is obtained from $P_1$ by flipping $B$ votes of $R$,  where $r(P_2)\succ_R r(P_1)$, then $\HM_B$ fails at $P_1$. 
\end{proof}
%\begin{proof}
%The proof is similar to the proof of Theorem~\ref{thm:CC+Par}. The only difference is that Claim~\ref{claim:tree-violation} is replaced by the following claim:
%
%{\bf Claim~\ref{claim:tree-violation}$^*$}{\em
%For any $P\in \calP_{n}$ and any violation tree $T\in \calT_{P,B}$, there exists a profile in $T$ where $\CC$ and/or $\HM_B$  is violated.
%}
%
%The proof of Claim~\ref{claim:tree-violation}$^*$ is straightforward: suppose there exists a pair of profiles $P_1, P_2$ on a violation tree, such that $P_2$ is obtained from $P_1$ by flipping $B$ votes of $R$,  where $r(P_2)\succ_R r(P_1)$, then $\PRP$ fails at $P_1$. 
%\end{proof} 

%Because every violation of $\HM$ is also a violation  $\SP$, the semi-random impossibility for $\CnS$ follows as a corollary.% of Theorem~\ref{thm:CC+HM}. 

\begin{thm}[\bf $\CC$+Maskin monotonicity]
\label{thm:CC+MM}
For any fixed $m\ge 3$, any  $\Pi$ that satisfies Assumption~\ref{asmp:strict},   any voting rule $r$, any $n\in\mathbb N$,  and any $1\le B\le \sqrt n$,  
$$\satmin{\CnM}{\Pi}(r,n, B) = 1- \Omega\left(\frac{B}{\sqrt n}\right)$$
\end{thm}
%The full proofs and the violation template (Figure~\ref{fig:VD-SP}) can be found in Appendix~\ref{app:proof-CC+MM}. 
%The theorem is proved by leveraging a simple proof diagram rooted at profiles that have a (Condorcet) cycle over $3$ alternatives in the WMG, where votes are changed to make certain alternatives the Condorcet winner in the leaves. The full proof can be found in Appendix~\ref{app:proof-CC+MM}. 

\begin{proof}
The proof follows the same pattern as that of the proof of Theorem~\ref{thm:CC+Par}. For simplicity, we present the proof for the case where  $m=3$, $\sqrt n$ is an integer,  $B\mid \sqrt n$, and $m!\mid n$. %The general case can be proved following a similar argument as in  the proof of Theorem~\ref{thm:CC+Par}.  
The main difference is that we use a different proof diagram whose corresponding violation template is illustrated in Figure~\ref{fig:VD-SP}.

\begin{figure}[htp]
\centering  
\includegraphics[width = .7\textwidth]{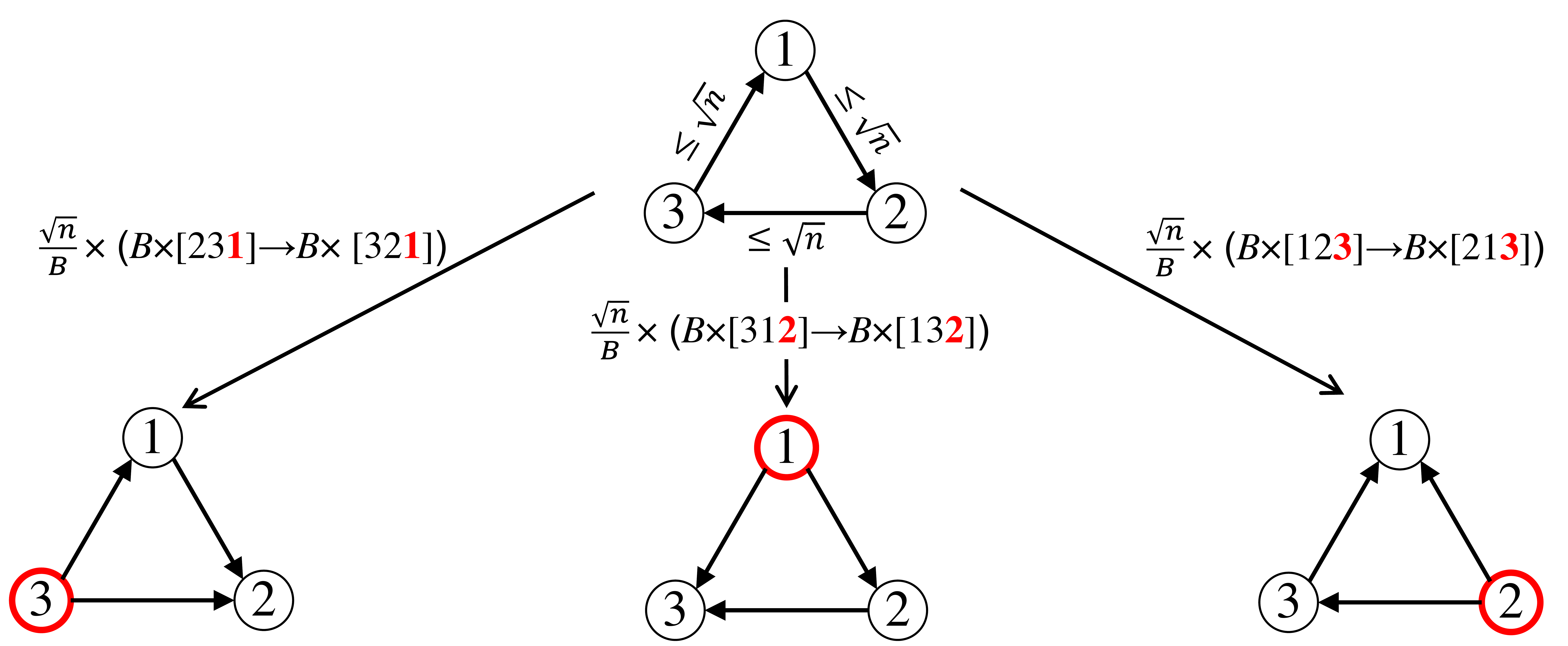}
\caption{Violation template for $\CnM$.\label{fig:VD-SP}}
\end{figure}

Formally, there are three branch starting from the root, each of which consists of $\frac{\sqrt n}{B}$ operations, and each operation consists of $B$ voters of the same vote $R_1$ changing to the same vote $R_2$ collaboratively: \begin{itemize}
\item for the first branch $R_1 = [231]$ and $R_2 = [321]$; 
\item for the second branch $R_1 = [312]$ and $R_2 = [132]$;  
\item for the third branch $R_1 = [123]$ and $R_2 = [213]$.
\end{itemize}
Let $\calP_n^\text{\CnM}$ denote the set of $n$-profiles $P$, such that the weights on $1\ra 2$, $2\ra 3$, and $3\ra 1$ in $P$'s WMG are between $0$ and $\sqrt n$. The violation template for $\CnM$ will be applied to profiles in $\calP_n^\text{\CnM}$, leading to violation trees for $\CnH$, denoted by $\calT_{n,B}^{\CnM}$. We then prove the following claim.
\begin{claimnew}  For every $P\in \calP_{n}^{\CnM}$ and every  $T\in \calT_{P,B}^{\CnM}$,  $\CC$ or $\MM_B$ is violated in $T$.
\end{claimnew}
To see why this claim holds, suppose $1$ is the winner of the root profile. We consider the left branch of the violation template in Figure~\ref{fig:VD-SP}. If the winner changes after any operation (which consists of $B$ votes of $[231]$ changing to $[321]$), then $\MM_B$ is violated. If $\MM_B$ is never violated on the left branch, then $1$ is the winner in the leaf node, which violates $\CC$ as $3$ is the Condorcet winner. The proof for $2$ (respective, $3$) being the winner is similar, by considering the middle (respectively, right) branch of Figure~\ref{fig:VD-SP}.

The rest of the proof, including the extension of the proof to the general case, is similar to the proof of Theorem~\ref{thm:CC+Par}.
%Because the violation template has $\Theta(\frac{\sqrt n}{B})$ nodes, and each of which modifies  $O(\sqrt n)$ votes, we can prove that  $\satmax{\neg(\CnM)}{\Pi}$ is $\Omega(\frac{B}{\sqrt n})$ in a similar way.
\end{proof}

%\begin{proof}
%We present the proof for the case where  $m=3$, $\sqrt n$ is an integer,  $B\mid \sqrt n$, and $m!\mid n$. The general case can be proved following a similar argument as in  the proof of Theorem~\ref{thm:CC+Par}. The proof follows the same pattern as that of the proof of Theorem~\ref{thm:CC+Par}. Let $\calP_n^\text{\CnM}$ denote the set of $n$-profiles $P$, such that the weights on $1\ra 2$, $2\ra 3$, and $3\ra 1$ are between $0$ and $\sqrt n$. Then, we define the violation template as illustrated in Figure~\ref{fig:VD-SP} and consider violation trees rooted in $\calP_n^\text{\CnM}$. Because the violation template has $\Theta(\frac{\sqrt n}{B})$ nodes, and each of which modifies  $O(\sqrt n)$ votes, we can prove that  $\satmax{\neg(\CnM)}{\Pi}$ is $\Omega(\frac{B}{\sqrt n})$ in a similar way.
%\end{proof}

\begin{thm}[\bf $\CC$+strategy-proofness]
\label{thm:CC+SP}
For any fixed $m\ge 3$, any  $\Pi$ that satisfies Assumption~\ref{asmp:strict},  any voting rule $r$, any $n\in\mathbb N$,  and any $1\le B\le \sqrt n$,  
$$\satmin{\CnS}{\Pi}(r,n, B) = 1- \Omega\left(\frac{B}{\sqrt n}\right)$$
\end{thm}
\begin{proof}  When $m\ge 4$, Theorem~\ref{thm:CC+SP} follows after Theorem~\ref{thm:CC+HM}, as $\SP$ is stronger than $\HM$. The proof (for all $m\ge 3$)   uses the same violation template in the proof of Theorem~\ref{thm:CC+MM} (Figure~\ref{fig:VD-SP}). The main difference is to prove that each violation tree generated by the violation template contains a profile where $\CC$ or $\SP_B$ is violated, i.e., the following claim. 
\begin{claimnew}  For every $P\in \calP_{n}^{\CnM}$ and every  $T\in \calT_{P,B}^{\CnM}$,  $\CC$ or $\SP_B$ is violated in $T$.
\end{claimnew}
To see why the claim holds, notice that for any profile $P$ that is not a leaf in Figure~\ref{fig:VD-SP}, if the winner is different from the highlighted alternative, then $B$ voters have incentive to improve the winner at the previous profile. If $\SP_B$ is never violated at all profiles on the tree, then $\CC$ is violated at a leaf node. 
\end{proof}

Recall that IC corresponds to $\Pi=\{\piuni\}$. Therefore, all semi-random impossibilities in this paper hold for IC. %, which is formally state as the following corollary.
\begin{coro}[\bf Quantitative Impossibilities under IC]
\label{coro:IC}
For any $X\in \{\CnP,\CnH, \CnM, \CnS\}$, any $m\ge 4$ ($m\ge 3$ for $\CnM$ and $\CnS$),   any voting rule $r$, any sufficiently large $n\in\mathbb N$,  and any $1\le B\le \sqrt n$,  

$\hfill\Pr\nolimits_{P\sim\text{IC}} X(r,P, B) = 1- \Omega\left(\frac{B}{\sqrt n}\right)\hfill$
\end{coro}
%As another example, let $\Pi$ denote the set of distributions in Example~\ref{ex:Pi}. Even though $\piuni\notin\Pi$, we have $\piuni = \frac 12 \cdot \pi^1+\frac 12 \cdot \pi^2\in\conv(\Pi)$. Therefore,  Theorem~\ref{thm:smoothed-imp-Con-Par} can be applied.
Corollary~\ref{coro:IC} also implies that for any voting rule that satisfies $\CC$, the likelihood for $r$ to satisfy 
$\Par$, $\HM$, $\MM$, or $\SP$, respectively, is   $\Omega\left(\frac{B}{\sqrt n}\right)$ under IC.

\section{Summary and Future Work}
We prove the first set of semi-random impossibility results involving $\CC$, showing that many existing voting rules are already optimal for $\CnP$ (for $B=1$) and $\CnS$ (for every $B\le \sqrt n$). The proof technique has potential to strengthen other worst-case impossibilities to their semi-random variants. For future work, we conjecture that all bounds for the axioms are tight and can be achieved by many rules that satisfy $\CC$. Other promising directions include  addressing the limitations discussed in Section~\ref{sec:related} and proving semi-random variants of other worst-case impossibility results, such as  Arrow's, Gibbard-Satterthwait (for non-$\CC$ rules), and various impossibility theorems in judgement aggregation. The proof technique developed in this paper (see Section~\ref{sec:other-imp}) does not seem to be directly applicable, because existing proofs use diagrams that contains $\Theta(n)$ nodes.

%We prove a smoothed impossibility theorem on {\sc Condorcet Criterion} and {\sc Participation}. An immediate open  question is to close the gap between the $n^{-3}$ rate proved in the theorem and the best known rate of $n^{-0.5}$ of existing voting rules. The problem appears challenging even under IC. Another open question is to characterize the the smoothed likelihood of $\nCP$ impossibility for the max-adversary, especially for $\Pi$'s such that $\piuni\notin\Pi$. More generally, we believe that developing smoothed versions of other impossibility theorems, especially those where quantitative versions were obtained under IC, such as Arrow's, Gibbard-Satterthwait, and various impossibility theorems in judgement aggregation, is a promising direction for future work. 
%\bibliographystyle{splncsnat}
%\bibliographystyle{plainnat}
%\newpage
{%\small
\bibliographystyle{plainnat}
\bibliography{/Users/administrator/GGSDDU/references}

\begin{thebibliography}{58}
\providecommand{\natexlab}[1]{#1}
\providecommand{\url}[1]{\texttt{#1}}
\expandafter\ifx\csname urlstyle\endcsname\relax
  \providecommand{\doi}[1]{doi: #1}\else
  \providecommand{\doi}{doi: \begingroup \urlstyle{rm}\Url}\fi

\bibitem[Arrow(1963)]{Arrow63:Social}
Kenneth Arrow.
\newblock \emph{Social choice and individual values}.
\newblock New Haven: Cowles Foundation, 2nd edition, 1963.
\newblock 1st edition 1951.

\bibitem[Bai et~al.(2022)Bai, Feige, G{\"o}lz, and Procaccia]{Bai2022:Fair}
Yushi Bai, Uriel Feige, Paul G{\"o}lz, and Ariel~D. Procaccia.
\newblock {Fair Allocations for Smoothed Utilities}.
\newblock In \emph{Proceedings of ACM EC}, 2022.

\bibitem[Baumeister et~al.(2020)Baumeister, Hogrebe, and
  Rothe]{Baumeister2020:Towards}
Dorothea Baumeister, Tobias Hogrebe, and J\"org Rothe.
\newblock {Towards Reality: Smoothed Analysis in Computational Social Choice}.
\newblock In \emph{Proceedings of AAMAS}, pages 1691--1695, 2020.

\bibitem[Berry et~al.(1995)Berry, Levinsohn, and Pakes]{Berry95:Automobile}
Steven Berry, James Levinsohn, and Ariel Pakes.
\newblock Automobile prices in market equilibrium.
\newblock \emph{Econometrica}, 63\penalty0 (4):\penalty0 841--890, 1995.

\bibitem[Blum and Dunagan(2002)]{Blum2002:Smoothed}
Avrim Blum and John~D Dunagan.
\newblock {Smoothed Analysis of the Perceptron Algorithm for Linear
  Programming}.
\newblock In \emph{Proceedings of SODA}, pages 905--914, 2002.

\bibitem[Blum and G{\"o}lz(2021)]{Blum2021:Incentive-Compatible}
Avrim Blum and Paul G{\"o}lz.
\newblock {Incentive-Compatible Kidney Exchange in a Slightly Semi-Random
  Model}.
\newblock In \emph{Proceedings of ACM EC}, 2021.

\bibitem[Blum and Spencer(1995)]{Blum1995:Coloring}
Avrim Blum and Joel Spencer.
\newblock {Coloring Random and Semi-Random k-Colorable Graphs}.
\newblock \emph{Journal of Algorithms}, 19\penalty0 (2):\penalty0 204--234,
  1995.

\bibitem[Boodaghians et~al.(2020)Boodaghians, Brakensiek, Hopkins, and
  Rubinstein]{Boodaghians2020:Smoothed}
Shant Boodaghians, Joshua Brakensiek, Samuel~B. Hopkins, and Aviad Rubinstein.
\newblock {Smoothed Complexity of 2-player Nash Equilibria}.
\newblock In \emph{Proceedings of FOCS}, 2020.

\bibitem[Brandt et~al.(2017)Brandt, Geist, and Peters]{Brandt2017:Optimal}
Felix Brandt, Christian Geist, and Dominik Peters.
\newblock {Optimal bounds for the no-show paradox via SAT solving}.
\newblock \emph{Mathematical Social Sciences}, 90:\penalty0 18--27, 2017.

\bibitem[Brandt et~al.(2021)Brandt, Hofbauer, and
  Strobel]{Brandt2021:Exploring}
Felix Brandt, Johannes Hofbauer, and Martin Strobel.
\newblock {Exploring the No-Show Paradox for Condorcet Extensions}.
\newblock In Mostapha Diss and Vincent Merlin, editors, \emph{{Evaluating
  Voting Systems with Probability Models}}. Springer, 2021.

\bibitem[Chung et~al.(2008)Chung, Ligett, Pruhs, and Roth]{Chung2008:The-Price}
Christine Chung, Katrina Ligett, Kirk Pruhs, and Aaron Roth.
\newblock {The Price of Stochastic Anarchy}.
\newblock In \emph{International Symposium on Algorithmic Game Theory}, pages
  303--314, 2008.

\bibitem[Condorcet(1785)]{Condorcet1785:Essai}
Marquis~de Condorcet.
\newblock \emph{Essai sur l'application de l'analyse \`a la probabilit\'e des
  d\'ecisions rendues \`a la pluralit\'e des voix}.
\newblock Paris: L'Imprimerie Royale, 1785.

\bibitem[Diss and Merlin(2021)]{Diss2021:Evaluating}
Mostapha Diss and Vincent Merlin, editors.
\newblock \emph{{Evaluating Voting Systems with Probability Models}}.
\newblock Studies in Choice and Welfare. Springer International Publishing,
  2021.

\bibitem[Dobzinski and Procaccia(2008)]{Dobzinski08:Frequent}
Shahar Dobzinski and Ariel~D. Procaccia.
\newblock Frequent manipulability of elections: The case of two voters.
\newblock In \emph{Proceedings of the Fourth Workshop on Internet and Network
  Economics (WINE)}, pages 653--664, Shanghai, China, 2008.

\bibitem[Feige(2021)]{Feige2021:Introduction}
Uriel Feige.
\newblock {Introduction to Semi-Random Models}.
\newblock In Tim Roughgarden, editor, \emph{Beyond the Worst-Case Analysis of
  Algorithms}. Cambridge University Press, 2021.

\bibitem[Filmus et~al.(2020)Filmus, Lifshitz, Minzer, and
  Mossel]{Filmus2020:AND-testing}
Yuval Filmus, Noam Lifshitz, Dor Minzer, and Elchanan Mossel.
\newblock {AND testing and robust judgement aggregation}.
\newblock In \emph{Proceedings of STOC}, 2020.

\bibitem[Fishburn(1974{\natexlab{a}})]{Fishburn1974:Aspects}
Peter~C. Fishburn.
\newblock {Aspects of One-Stage Voting Rules}.
\newblock \emph{Management Science}, 21\penalty0 (4):\penalty0 422--427,
  1974{\natexlab{a}}.

\bibitem[Fishburn(1974{\natexlab{b}})]{Fishburn1974:Simple}
Peter~C. Fishburn.
\newblock {Simple voting systems and majority rule}.
\newblock \emph{Behavioral Science}, 19\penalty0 (3):\penalty0 166--176,
  1974{\natexlab{b}}.

\bibitem[Fishburn(1974{\natexlab{c}})]{Fishburn74:Paradoxes}
Peter~C. Fishburn.
\newblock Paradoxes of voting.
\newblock \emph{The American Political Science Review}, 68\penalty0
  (2):\penalty0 537--546, 1974{\natexlab{c}}.

\bibitem[Fishburn and Brams(1983)]{Fishburn1983:Paradoxes}
Peter~C. Fishburn and Steven~J. Brams.
\newblock {Paradoxes of Preferential Voting}.
\newblock \emph{Mathematics Magazine}, 56\penalty0 (4):\penalty0 207--214,
  1983.

\bibitem[Friedgut et~al.(2011)Friedgut, Kalai, Keller, and
  Nisan]{Friedgut2011:A-quantitative}
Ehud Friedgut, Gil Kalai, Nathan Keller, and Noam Nisan.
\newblock {A Quantitative Version of the Gibbard-Satterthwaite theorem for
  Three Alternatives}.
\newblock \emph{SIAM Journal on Computing}, 40\penalty0 (3):\penalty0 934--952,
  2011.

\bibitem[Gehrlein and Fishburn(1978)]{Gehrlein1978:Coincidence}
William~V. Gehrlein and Peter~C. Fishburn.
\newblock {Coincidence probabilities for simple majority and positional voting
  rules}.
\newblock \emph{Social Science Research}, 7\penalty0 (3):\penalty0 272--283,
  1978.

\bibitem[Gehrlein and Lepelley(2011)]{Gehrlein2011:Voting}
William~V. Gehrlein and Dominique Lepelley.
\newblock \emph{{Voting Paradoxes and Group Coherence: The Condorcet Efficiency
  of Voting Rules}}.
\newblock Springer, 2011.

\bibitem[Gibbard(1973)]{Gibbard73:Manipulation}
Allan Gibbard.
\newblock Manipulation of voting schemes: {A} general result.
\newblock \emph{Econometrica}, 41:\penalty0 587--601, 1973.

\bibitem[Isaksson et~al.(2010)Isaksson, Kindler, and
  Mossel]{Isaksson10:Geometry}
Marcus Isaksson, Guy Kindler, and Elchanan Mossel.
\newblock {The Geometry of Manipulation: A Quantitative Proof of the
  Gibbard-Satterthwaite Theorem}.
\newblock In \emph{Proceedings of the 51st Annual Symposium on Foundations of
  Computer Science (FOCS)}, pages 319--328, Washington, DC, USA, 2010.

\bibitem[Kalai(2002)]{Kalai02:Fourier}
Gil Kalai.
\newblock A {F}ourier-theoretic perspective on the {C}ondorcet paradox and
  {A}rrow's theorem.
\newblock \emph{Advances in Applied Mathematics}, 29\penalty0 (3):\penalty0
  412---426, 2002.

\bibitem[Keller(2012)]{Keller2012:-A-tight}
Nathan Keller.
\newblock { A tight quantitative version of Arrow's impossibility theorem}.
\newblock \emph{Journal of the European Mathematical Society}, 14:\penalty0
  1331--1355, 2012.

\bibitem[List and Pettit(2002)]{List2002:Aggregating}
Christian List and Philip Pettit.
\newblock {Aggregating sets of judgments: An impossibility result}.
\newblock \emph{Economics and philosophy}, 18\penalty0 (1):\penalty0 89--110,
  2002.

\bibitem[List and Pettit(2004)]{List2004:Aggregating}
Christian List and Philip Pettit.
\newblock {Aggregating Sets of Judgments: Two Impossibility Results Compared}.
\newblock \emph{Synthese}, 140:\penalty0 207--235, 2004.

\bibitem[Liu and Xia(2022)]{Liu2022:Semi-Random}
Ao~Liu and Lirong Xia.
\newblock {The Semi-Random Likelihood of Doctrinal Paradoxes}.
\newblock In \emph{Proceedings of AAAI}, 2022.

\bibitem[Maskin(1999)]{Maskin1999:Nash}
Eric Maskin.
\newblock {Nash Equilibrium and Welfare Optimality}.
\newblock \emph{Review of Economic Studies}, 66:\penalty0 23---38, 1999.

\bibitem[Mossel(2012)]{Mossel2012:A-quantitative}
Elchanan Mossel.
\newblock {A quantitative Arrow theorem}.
\newblock \emph{Probability Theory and Related Fields}, 154:\penalty0 49--88,
  2012.

\bibitem[Mossel and Racz(2015)]{Mossel2015:A-quantitative}
Elchanan Mossel and Miklos~Z. Racz.
\newblock {A quantitative Gibbard-Satterthwaite theorem without neutrality}.
\newblock \emph{Combinatorica}, 35\penalty0 (3):\penalty0 317--387, 2015.

\bibitem[Mossel et~al.(2013)Mossel, Oleszkiewicz, and
  Sen]{Mossel2013:On-reverse}
Elchanan Mossel, Krzysztof Oleszkiewicz, and Arnab Sen.
\newblock {On reverse hypercontractivity}.
\newblock \emph{Geometric and Functional Analysis}, 23\penalty0 (3):\penalty0
  1062--1097, 2013.

\bibitem[Moulin(1983)]{Moulin1983:The-Strategy}
Herv{\'e} Moulin.
\newblock \emph{{The Strategy of Social Choice}}.
\newblock Elsevier, 1983.

\bibitem[Moulin(1988)]{Moulin1988:Condorcets}
Herv{\'e} Moulin.
\newblock {Condorcet's principle implies the no show paradox}.
\newblock \emph{Journal of Economic Theory}, 45\penalty0 (1):\penalty0 53--64,
  1988.

\bibitem[Muller and Satterthwaite(1977)]{Muller77:SPA}
E.~Muller and Mark Satterthwaite.
\newblock The equivalence of strong positive association and
  strategy-proofness.
\newblock \emph{Journal of Economic Theory}, 14:\penalty0 412--418, 1977.

\bibitem[Nehama(2013)]{Nehama2013:Approximately}
Ilan Nehama.
\newblock Approximately classic judgement aggregation.
\newblock \emph{Annals of Mathematics and Artificial Intelligence},
  68:\penalty0 91--134, 2013.

\bibitem[Newenhizen(1992)]{Newenhizen1992:The-Borda}
Jill~Van Newenhizen.
\newblock {The Borda Method Is Most Likely to Respect the Condorcet Principle}.
\newblock \emph{Economic Theory}, 2\penalty0 (1):\penalty0 69--2--83, 1992.

\bibitem[Paris(1975)]{Paris1975:Plurality}
David~C. Paris.
\newblock {Plurality distortion and majority rule}.
\newblock \emph{Behavioral Science}, 20\penalty0 (2):\penalty0 125--133, 1975.

\bibitem[Pattanaik(1978)]{Pattanaik1978:Strategy}
Prasanta~K. Pattanaik.
\newblock \emph{{Strategy and group choice}}.
\newblock Elsevier North-Holland, 1978.

\bibitem[Peters(2017)]{Peters2017:Condorcets}
Dominik Peters.
\newblock {Condorcet's principle and the preference reversal paradox}.
\newblock In \emph{Proceedings of TARK}, 2017.

\bibitem[Peters(2019)]{Peters2019:Fair}
Dominik Peters.
\newblock \emph{{Fair Division of the Commons}}.
\newblock PhD thesis, Oxford University, 2019.

\bibitem[Psomas et~al.(2019)Psomas, Schvartzman, and
  Weinberg]{Psomas2019:Smoothed}
Alexandros Psomas, Ariel Schvartzman, and Matthew~S. Weinberg.
\newblock {Smoothed Analysis of Multi-Item Auctions with Correlated Values}.
\newblock In \emph{Proceedings of ACM EC}, 2019.

\bibitem[Roughgarden(2021)]{Roughgarden2021:Beyond}
Tim Roughgarden.
\newblock \emph{Beyond the Worst-Case Analysis of Algorithms}.
\newblock Cambridge University Press, 2021.

\bibitem[Saari(1995)]{Saari1995:Basic}
Donald~G. Saari.
\newblock \emph{{Basic Geometry of Voting}}.
\newblock Springer, 1995.

\bibitem[Sanver and Zwicker(2009)]{Sanver2009:One-way}
M.~Remzi Sanver and William~S. Zwicker.
\newblock {One-way monotonicity as a form of strategy-proofness}.
\newblock \emph{International Journal of Game Theory}, 38:\penalty0 553--574,
  2009.

\bibitem[Satterthwaite(1975)]{Satterthwaite75:Strategy}
Mark Satterthwaite.
\newblock Strategy-proofness and {A}rrow's conditions: Existence and
  correspondence theorems for voting procedures and social welfare functions.
\newblock \emph{Journal of Economic Theory}, 10:\penalty0 187--217, 1975.

\bibitem[Spielman and Teng(2004)]{Spielman2004:Smoothed}
Daniel~A. Spielman and Shang-Hua Teng.
\newblock {Smoothed analysis of algorithms: Why the simplex algorithm usually
  takes polynomial time}.
\newblock \emph{Journal of the ACM}, 51\penalty0 (3), 2004.

\bibitem[Spielman and Teng(2009)]{Spielman2009:Smoothed}
Daniel~A. Spielman and Shang-Hua Teng.
\newblock {Smoothed Analysis: An Attempt to Explain the Behavior of Algorithms
  in Practice}.
\newblock \emph{Communications of the ACM}, 52\penalty0 (10):\penalty0 76--84,
  2009.

\bibitem[Thurstone(1927)]{Thurstone27:Law}
Louis~Leon Thurstone.
\newblock A law of comparative judgement.
\newblock \emph{Psychological Review}, 34\penalty0 (4):\penalty0 273--286,
  1927.

\bibitem[Train(2009)]{Train09:Discrete}
Kenneth~E. Train.
\newblock \emph{{Discrete Choice Methods with Simulation}}.
\newblock Cambridge University Press, 2nd edition, 2009.

\bibitem[Xia(2020)]{Xia2020:The-Smoothed}
Lirong Xia.
\newblock {The Smoothed Possibility of Social Choice}.
\newblock In \emph{Proceedings of NeurIPS}, 2020.

\bibitem[Xia(2021{\natexlab{a}})]{Xia2021:How-Likely}
Lirong Xia.
\newblock {How Likely Are Large Elections Tied?}
\newblock In \emph{Proceedings of ACM EC}, 2021{\natexlab{a}}.

\bibitem[Xia(2021{\natexlab{b}})]{Xia2021:Semi-Random}
Lirong Xia.
\newblock {The Semi-Random Satisfaction of Voting Axioms}.
\newblock In \emph{Proceedings of NeurIPS}, 2021{\natexlab{b}}.

\bibitem[Xia(2022)]{Xia2022:How-Likely}
Lirong Xia.
\newblock {How Likely A Coalition of Voters Can Influence A Large Election?}
\newblock \emph{arXiv:2202.06411}, 2022.

\bibitem[Xia and Conitzer(2008)]{Xia08:Sufficient}
Lirong Xia and Vincent Conitzer.
\newblock A sufficient condition for voting rules to be frequently manipulable.
\newblock In \emph{Proceedings of the ACM Conference on Electronic Commerce
  (EC)}, pages 99--108, Chicago, IL, USA, 2008.

\bibitem[Xia and Zheng(2021)]{Xia2021:The-Smoothed}
Lirong Xia and Weiqiang Zheng.
\newblock The smoothed complexity of computing kemeny and slater rankings.
\newblock In \emph{Proceedings of AAAI}, 2021.

\end{thebibliography}
}

\end{document}